\documentclass{article} 
\usepackage[nonatbib, preprint]{neurips_2020}
\usepackage[utf8]{inputenc}
\usepackage{amsmath, amsthm, xcolor, algorithm, url, amssymb}
\usepackage{float}
\usepackage{mathtools}
\usepackage[noend]{algpseudocode}
\usepackage[xcolor]{changes}
\usepackage{glossaries}
\usepackage[hidelinks,breaklinks=true]{hyperref}
\usepackage{varwidth}

\usepackage{tikz}
\usepackage{graphicx}
\usetikzlibrary{matrix}
\usetikzlibrary{trees}
\usetikzlibrary{arrows.meta}
\usetikzlibrary{positioning, fit}
\usetikzlibrary{shapes, decorations}
\usetikzlibrary{decorations.pathreplacing}
\hypersetup{
    colorlinks,
    citecolor=blue,
    filecolor=black,
    linkcolor=black,
    urlcolor=black
}

\theoremstyle{definition}

\newtheorem{theorem}{Theorem}[section]

\newtheorem{lemma}[theorem]{Lemma}
\newtheorem{prop}[theorem]{Proposition}
\newtheorem{define}[theorem]{Definition}
\newtheorem{ex}[theorem]{Example}
\newtheorem*{remark}{Remark}

\numberwithin{equation}{section}
\numberwithin{algorithm}{section}

\DeclareMathOperator*{\argmax}{arg\,max}

\newcommand{\prob}[1]{\mathbb{P}\left ( #1 \right)} 
\newcommand{\e}[1]{\exp \left( #1 \right)} 
\newcommand{\E}[1]{\mathbb{E}\left[ #1 \right]} 

\newcommand{\validatorset}{\mathcal{V}}

\newcommand{\gblock}{B_{\text{genesis}}}
\newcommand{\mainprotocol}{\mathsf{Gasper}}

\newcommand{\eps}{\epsilon}

\newcommand{\mlmd}{\mathsf{HLMD}}
\newcommand{\ebb}{\mathsf{EBB}}
\newcommand{\lastepoch}{\mathsf{LEBB}}
\newcommand{\lastepochpair}{\mathsf{LE}}

\newcommand{\lastjustpair}{\mathsf{LJ}}
\newcommand{\block}{\mathsf{block}}
\newcommand{\fork}{\mathsf{fork}}
\newcommand{\slot}{\mathsf{slot}}
\newcommand{\epoch}{\mathsf{ep}}
\newcommand{\attestepoch}{\mathsf{aep}}
\newcommand{\chain}{\mathsf{chain}}
\newcommand{\fview}{\mathsf{fv}}
\newcommand{\view}{\mathsf{view}}
\newcommand{\ffgview}{\mathsf{ffgview}}

\newcommand{\network}{\mathsf{NW}}
\newcommand{\newattests}{\mathsf{newattests}}
\newcommand{\votes}{\xrightarrow[]{\mathsf{V}}}
\newcommand{\supervotes}{\xrightarrow[]{\mathsf{J}}}

\title{Combining GHOST and Casper}

\author{%
  Vitalik Buterin \\
  Ethereum Foundation \\
  \texttt{v@buterin.com} \\
  \And
  Diego Hernandez \\
  SJSU \\
  \texttt{diego.hernandez@sjsu.edu} \\
  \AND
  Thor Kamphefner \\
  SJSU \\
  \texttt{thor.kamphefner@sjsu.edu} \\
  \And
  Khiem Pham \\
  VinAI Research \\
  \texttt{v.khiempd1@vinai.io} \\
  \And
  Zhi Qiao \\
  SJSU \\
  \texttt{qiaozhi9092@gmail.com} \\
  \And
  Danny Ryan \\
  Ethereum Foundation \\
  \texttt{danny@ethereum.org} \\
  \And
  Juhyeok Sin \\
  SJSU \\
  \texttt{homin203@gmail.com} \\
  \And
  Ying Wang \\
  SJSU \\
  \texttt{yingwang1001@gmail.com} \\
  \And
  Yan X Zhang \\
  SJSU \\
  \texttt{yan.x.zhang@sjsu.edu} \\
}
\date{\today}

\begin{document}

\maketitle

\begin{abstract}
We present ``Gasper,'' a proof-of-stake-based consensus protocol, which is an idealized version of the proposed Ethereum 2.0 beacon chain. The protocol combines Casper FFG, a finality tool, with LMD GHOST, a fork-choice rule. We prove safety, plausible liveness, and probabilistic liveness under different sets of assumptions.
\end{abstract}

\tableofcontents

\newpage

\section{Introduction}
Our goal is to create a consensus protocol for a proof-of-stake blockchain. This paper is motivated by the need to construct Ethereum's ``beacon chain'' in its sharding design. However, the construction can also be used for a solo blockchain with no sharding. 

We present $\mainprotocol$, a combination of Casper FFG (the Friendly Finality Gadget) and LMD GHOST (Latest Message Driven Greediest Heaviest Observed SubTree). Casper FFG \cite{buterin2017casper} is a ``finality gadget,'' which is an algorithm that marks certain blocks in a blockchain as \emph{finalized} so that participants with partial information can still be fully confident that the blocks are part of the canonical chain of blocks. Casper is not a fully-specified protocol and is designed to be a ``gadget'' that works on top of a provided blockchain protocol, agnostic to whether the provided chain were proof-of-work or proof-of-stake. LMD GHOST is a \emph{fork-choice rule} where \emph{validators} (participants) \emph{attest} to blocks to signal support for those blocks, like voting. $\mainprotocol$ is a full proof-of-stake protocol that is an idealized abstraction of the proposed Ethereum implementation.


In Sections~\ref{sec:setup} and ~\ref{sec:stage-1}, we define our primitives, provide background knowledge, and state our goals, In Section~\ref{sec:stage-2}, we give our main protocol $\mainprotocol$. In Sections~\ref{sec:safety}, \ref{sec:plausible-liveness}, and \ref{sec:liveness}, we formally prove $\mainprotocol$'s desired qualities. In Section~\ref{sec:practice-vs-theory}, we summarize some differences between $\mainprotocol$ and the actual planned Ethereum implementation, such as delays in accepting attestations, delaying finalization, and dynamic validator sets. We conclude with some thoughts and ideas for future research in Section~\ref{sec:conclusion}.



\section{Setup and Goals}
\label{sec:setup}

The main goal of this paper is to describe and prove properties of a proof-of-stake blockchain with certain safety and liveness claims. In this section, we give some general background on consensus protocols and blockchain. As the literature uses a diverse lexicon, the primary purpose of this section is to pick out a specific set of vocabulary for the rest of the paper.


\subsection{Consensus Protocols, Validators, Blockchain}

Our goal is to create a \emph{consensus protocol} (or \emph{protocol} for short from this point on), which is an agreed suite of algorithms for a set of entities (nodes, people, etc.) to follow in order to obtain a consensus history of their state, even if the network is unreliable and/or if many validators are malicious.

We call the entities (people, programs, etc.) who participate in protocol \emph{validators}, denoted by the set $\validatorset$. They are called ``validators'' because of their data-validation role in the Ethereum 2.0 beacon chain, though the specifics of this role are not part of our abstract protocol. Other works use many different words for the same concept: e.g. ``replicas'' in classic PBFT literature, ``block producers'' in DPoS and EOS, ``peers'' on Peercoin and Nxt, or ``bakers'' on Tezos. 
The validators are connected to each other on a (peer-to-peer) \emph{network}, which means they can broadcast \emph{messages} (basically, packets of data) to each other. In the types of protocols we are interested in, the primary types of messages are \emph{blocks}, which collect bundles of messages. The first block is called the \emph{genesis block} and acts as the initial ``blank slate'' state. Other blocks are descriptions of state transitions with a pointer to a \emph{parent} block\footnote{Our choice that each block has a unique parent is almost taken for granted in blockchains; one can imagine blockchains where blocks do not need to have parents, or can have more than one parent. See, for instance, Hashgraph \cite{hashgraph}. Even the very general structure of ``block'' itself may be limiting!}. A \emph{blockchain} is just an instantiation (i.e. with specific network conditions, validators, etc.) of such a consensus protocol where we use blocks to build the overall state.

We do not want all blocks seen on the network to be accepted as common history because blocks can conflict with each other. These conflicts can happen for honest reasons such as network latency, or for malicious reasons such as byzantine validators wanting to double-spend.  \emph{Byzantine} refers to a classic problem in distributed systems, the Byzantine Generals Problem, and is used in this instance to describe a validator who does not follow the protocol, either because they are malicious or experiencing network latency. Graph-theoretically, one can think of a choice of history as a choice of a ``chain'' going from the genesis block to a particular block. Thus, a consensus history is a consensus of which chain of blocks to accept as correct. This is why we call such an instantiation of our protocol a ``blockchain.''

\begin{ex}
\label{ex:yunice}
Assume the goal of a protocol is to keep a ledger of every validator's account balances. Then their collective state is the ledger, and each block captures a state transition, such as a monetary transaction from one validator to another (e.g., ``Yunice sends coin with ID $538$ to Bureaugard,''). One type of conflict we may want to avoid is having both the aforementioned transaction and another transaction with content ``Yunice sends coin with ID $538$ to Carol'' be accepted into consensus history, an act commonly called \emph{double-spending}. The current consensus state can be determined by starting from a ``genesis state'' and sequentially processing each message in the consensus history (hopefully one that all honest validators will agree with), starting from the genesis block.
\end{ex}

\subsection{Messages and Views}


In our model, a validator $V$ interacts with the network by broadcasting \emph{messages}, which are strings in some language. When an honest validator $V$ broadcasts a message $M$, we assume $M$ is sent to all validators on the network\footnote{We model the validators as sending a single message to an abstract ``network,'' which handles the message. In practice, it may be part of the protocol that e.g., every honest validator rebroadcasts all messages that they see, for robustness. We consider these implementation details not important to our paper.}. 

The main type of message \emph{proposes} a \emph{block}, which is a piece of data, to the network. Other messages can be  bookkeeping notices such as voting for blocks (``attestation''), putting new validators on the blockchain (``activation''), proving bad actions of other validators (``slashing''), etc. depending on the specific protocol. We assume that each message is digitally signed by a single validator, which means we can accurately trace the author of each message and attacks such as impersonation of honest validators are not possible. 

Due to network latency and dishonest validators (delaying messages or relaying incorrect information), validators may have different states of knowledge of the full set of messages given to the network; we formalize this by saying that a validator either \emph{sees} or \emph{does not see} each message given to the network at any given time.

Each message may have one or more \emph{dependencies}, where each dependency is another message. At any time, we \emph{accept} a message if and only if all of its dependencies (possibly none) are accepted, defined recursively. We now define the \emph{view} of a validator $V$ at a given time $T$, denoted as $\view(V, T)$, as the set of all the \textbf{accepted} messages the validator has seen so far. We also have a ``God's-eye-view'' that we call the \emph{network view}, defined to be the set of accepted messages for a hypothetical validator that has seen (with no latency) all messages any validator has broadcast at any time (this includes messages sent by a malicious validator to only a subset of the network). We will treat the network as a ``virtual validator'', so we use $\view(\network, T)$ to denote the network view at time $T$. For any validator $V$ and any given time $T$, $\view(\network, T)$ includes all the messages in $\view(V, T)$, though the timestamps may be mismatched. Finally, since usually the context is that we are talking about a specific point in time, we will usually just suppress the time and use notation such as $\view(V)$ to talk about $\view(V,T)$, unless talking about the specific time is necessary. 

\begin{ex}
At time $T$, suppose validator $V$ sees the message $M$ with content ``Yunice sends coin number $5340$ to Bob,'' but has not yet seen the message $M'$ containing ``Yunice obtains coin number $5340$'' (which is on the network but has not yet made its way to $V$). Also suppose that in our protocol $M$ depends on $M'$ (and no other message), which captures the semantic meaning that we need to obtain a coin before spending it, and that $V$ cannot and should not act on $M$ without seeing $M'$. In this situation we say that $V$ has seen but has not accepted $M$, so $M \notin \view(V,T)$. However, the network has seen both; so if the network accepts $M'$, we have $M \in \view(\network, T)$.
\end{ex}

We now make some assumptions on the structure of views:

\begin{itemize}
    \item Everyone (the validators and the network) starts with a single message (with no signature) corresponding to a single agreed-upon \emph{genesis block}, denoted $\gblock$, with no dependencies (so it starts out automatically accepted into all views).
    \item Each block besides $\gblock$ refers to (and depends on) a \emph{parent block} as part of its data. Thus, we can visualize $\view(V)$ as a directed acyclic graph (in fact, a directed tree\footnote{This statement already captures the primary reason we define both seeing and accepting; the blocks a validator \emph{accepts} must be a single rooted tree, though the blocks a validator \emph{sees} can be some arbitrary subset of the tree, so limiting our attention to the tree allows cleaner analysis and avoids pathological cases. Beyond this reason, the distinction between the two words is unimportant for the rest of this paper.}) rooted at $\gblock$, with an edge $B \leftarrow B'$ if $B$ is the parent of $B'$, in which case we say $B'$ is a \emph{child} of $B$. We call these \emph{parent-child} edges to differentiate from other types of edges.
    \item We say that a block is a \emph{leaf} block if it has no children.
    \item A \emph{chain} in such a protocol would then be a sequence of pairwise parent-child edges $B_1 \leftarrow B_2 \leftarrow \cdots$ If there is a chain from $B$ to $B'$, we say $B'$ is a \emph{descendant} of $B$. We say $B$ is an \emph{ancestor} of $B'$ if and only if $B'$ is a descendant of $B$. We say two blocks $B$ and $B'$ \emph{conflict} if they do not equal and neither is a descendant of the other. 
    \item The above setup implies each block $B$ uniquely defines a (backwards) chain starting from $\gblock$ to $B$, and this must be the same chain in any view that includes $B$ (thus we do not need to include a view as part of its definition). We call this \emph{the chain of $B$}, or $\chain(B)$.
\end{itemize}


\begin{figure}[H]
\begin{center}
\begin{tikzpicture}
\tikzset {
    basic/.style = {rectangle, minimum width=3em, minimum height = 1.3em},
    root/.style = {basic, thin, align=center},
    block/.style = {draw, basic, thin, align=center}
}
\node[draw, basic](A0) {{$\gblock$}};
\node[block, below of = A0](A) {{$A$}};
\node[block, below of = A](B){$B$};
\node[block, below left of = B](C){$C$};
\node[block, below right of = B](D){$D$};
\node[block, below of = C](E){$E$};

\draw[->] (A) -- (A0);
\draw[->] (B) -- (A);
\draw[->] (C) -- (B);
\draw[->] (D) -- (B);
\draw[->] (E) -- (C);



\tikzstyle{dash}=[dashed,fill=gray!50]
\end{tikzpicture}
\end{center}
\caption{A graph visualization of a view (only visualizing blocks and not other messages) that occurs in the types of protocols we care about.  \label{fig:view}} 
\end{figure}
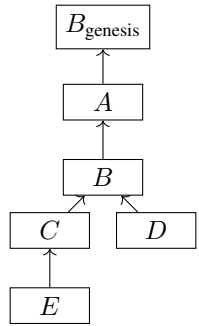

\begin{ex}
See Figure~\ref{fig:view} for an example of a view, limited to blocks. Arrows portray parent-child edges. The blocks in a view form a tree rooted at a single genesis block $\gblock$. Here blocks $C$ and $D$ conflict (as well as $E$ and $D$). $E$ has the longest chain, which is $\chain(E) = (\gblock, A, B, C, E)$. 
\end{ex}

There are some technicalities to our wording in this section that motivates our definition of \emph{views} but are not very important to the core of this paper. To not distract the reader, we defer these details to Appendix~\ref{sec:view}.

\subsection{Proof-of-stake}

The first and most influential approach to blockchain protocols is Bitcoin \cite{nakamoto2008bitcoin}, using a proof-of-work model. This approach makes it computationally difficult to propose blocks, using the simple and intuitive fork-choice rule, ``the block with the heaviest chain is the head of the chain." Miners (the analogy of validators) simply propose blocks to build on the heaviest chain (the chain with the most amount of computational work), which does not mathematically guarantee any non-genesis block as ``correct'' but offers a probabilistic guarantee since it becomes less and less likely to conflict with a block once other miners build on top of it. Then the consensus history is simply the chain with the heaviest amount of work, which ideally keeps growing. The elegance of proof-of-work is paid for by the expensive cost of computational work / electricity / etc. to propose a block. 

In this paper, we want to create a \emph{proof-of-stake} blockchain, where a validator's voting power is proportional to their bonded stake (or money) in the system. Instead of using computational power to propose blocks, proposing blocks is essentially free. In exchange, we need an additional layer of mathematical theory to prevent perverse incentives 
that arise when we make proposing blocks ``easy.'' We now shift our paper's focus to a family of blockchain designs with proof-of-stake in mind.


To start, we assume we have a set of $N$ \emph{validators} $\validatorset = \left\{V_1, \ldots, V_N\right\}$ and that each validator $V \in \validatorset$ has an amount of \emph{stake} $w(V)$, a positive real number describing an amount of collateral. We make the further assumption that the average amount of stake for each validator is $1$ unit, so the sum of the total stake is $N$. All of our operations involving stake will be linear, so this scaling does not change the situation and makes the bookkeeping easier. We assume validator sets and stake sizes are fixed for the sake of focusing on the theoretical essentials of the protocol. We address issues that occur when we relax these conditions in practical use in later sections, especially Section~\ref{sec:practice-vs-theory}.

The main ingredients we need are:
\begin{itemize}
    \item A \emph{fork-choice rule}:  a function $\fork()$ that, when given a view $G$, identifies a single leaf block $B$. This choice produces a unique chain $\fork(G) = \chain(B)$ from $\gblock$ to  $B$ called the \emph{canonical chain}. The block $B$ is called the \emph{head of the chain} in view $G$. Intuitively, a fork-choice rule gives a validator a ``law'' to follow to decide what the ``right'' block should be. For example, the \emph{longest chain rule} is a fork-choice rule that returns the leaf block which is farthest from the genesis block. This rule is similar with the ``heaviest chain rule'' used by Bitcoin.
    \item A concept of \emph{finality}: formally, a deterministic function $F$ that, when given a view $G$, returns a set $F(G)$ of \emph{finalized} blocks. Intuitively, finalized blocks are ``blocks that everyone will eventually think of as part of the consensus history" or ``what the blockchain is sure of.''
    \item \emph{Slashing conditions}: these are conditions that honest validators would never violate and violating validators can be provably caught, with the idea that violators' stake would be \emph{slashed}, or destroyed. Slashing conditions incentivize validators to follow the protocol (the protocol can reward honest validators for catching dishonest validators violating the conditions, for example). 
\end{itemize}

The key concept we use to define these ingredients are  \emph{attestations}, which are votes (embedded in messages) for which blocks ``should'' be the head of the chain. To prevent validators from double-voting or voting in protocol-breaking ways, we enforce \emph{slashing conditions} which can be used to destroy a validator's stake, incentivizing validators to follow the protocol. 


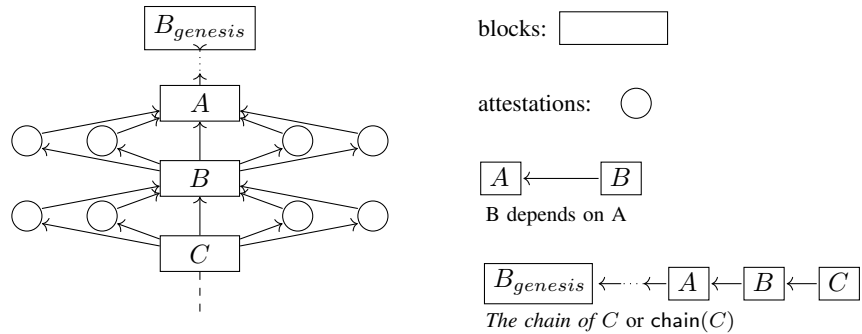
\begin{figure}[H]
    \begin{center}
\begin{tikzpicture}
\tikzset {
    basic/.style = {rectangle, minimum width=3em, minimum height = 1.3em},
    root/.style = {basic, thin, align=center},
    block/.style = {draw, basic, thin, align=center}
}
\node[draw, basic](A0) {{$B_{genesis}$}};
\draw[->] (A0) -- (0, -0.3);
\draw[dotted] (0, -0.3) -- (0, -0.7);
\draw[->] (A) -- (0, -0.6);
\node[block, below of = A0](A) {{$A$}};
\node[block, below of = A](B){$B$};
\node[block, below of = B](C){$C$};
\draw[dashed] (C) -- (0, -3.8);

\draw[] (-2.3,-1.5) circle(0.2cm);
\draw[] (-1.3,-1.5) circle(0.2cm);
\draw[] (1.3,-1.5) circle(0.2cm);
\draw[] (2.3,-1.5) circle(0.2cm);

\draw[] (-2.3,-2.5) circle(0.2cm);
\draw[] (-1.3,-2.5) circle(0.2cm);
\draw[] (1.3,-2.5) circle(0.2cm);
\draw[] (2.3,-2.5) circle(0.2cm);

\draw[->] (B) -- (A);
\draw[->] (C) -- (B);

\draw[->] (B) -- (-2.1, -1.6);
\draw[->] (B) -- (-1.1, -1.6);
\draw[->] (B) -- (1.1, -1.6);
\draw[->] (B) -- (2.1, -1.6);

\draw[->] (-2.1, -1.4) -- (A);
\draw[->] (-1.1, -1.4) -- (A);
\draw[->] (1.1, -1.4) -- (A);
\draw[->] (2.1, -1.4) -- (A);

\draw[->] (-2.1, -2.4) -- (B);
\draw[->] (-1.1, -2.4) -- (B);
\draw[->] (1.1, -2.4) -- (B);
\draw[->] (2.1, -2.4) -- (B);

\draw[->] (C) -- (-2.1, -2.6);
\draw[->] (C) -- (-1.1, -2.6);
\draw[->] (C) -- (1.1, -2.6);
\draw[->] (C) -- (2.1, -2.6);

\node[text width = 1.2cm, text height = 0.2cm, draw] at (5.5, 0){};
\node[text width = 2cm, font=\fontsize{9}{0}\selectfont] at (4.7, 0){blocks: };

\draw[] (5.8, -1) circle(0.2cm);
\node[text width = 2cm, font=\fontsize{9}{0}\selectfont] at (4.7, -1){attestations: };

\node[text width = 0.3cm, text height = 0.2cm, draw] at (5.6, -2){$B$};
\node[text width = 0.3cm, text height = 0.2cm, draw] at (4.0, -2){$A$};
\draw[<-] (4.3, -2) -- (5.3, -2);
\node[text width = 5cm, font=\fontsize{8}{0}\selectfont] at (6.3, -2.5){B depends on A};

\node[text width = 1.2cm, text height = 0.2cm, draw] at (4.5, -3.4){$B_{genesis}$};
\node[text width = 0.3cm, text height = 0.2cm, draw] at (6.5, -3.4){$A$};
\node[text width = 0.3cm, text height = 0.2cm, draw] at (7.5, -3.4){$B$};
\node[text width = 0.3cm, text height = 0.2cm, draw] at (8.5, -3.4){$C$};

\draw[->] (8.2, -3.4) -- (7.8, -3.4);
\draw[->] (7.2, -3.4) -- (6.8, -3.4);
\draw[->] (6.2, -3.4) -- (5.9, -3.4);
\draw[dotted] (5.9, -3.4) -- (5.6, -3.4);
\draw[->] (5.6, -3.4) -- (5.3, -3.4);
\node[text width = 5cm, font=\fontsize{8}{0}\selectfont] at (6.3, -3.9){\emph{The chain of $C$} or $\chain(C)$};

\end{tikzpicture}

\end{center}
     \caption{An ideal view with both blocks and attestations.}
\label{fig:attestations}
\end{figure}

One can see an idealized version of the protocol we will present in Figure~\ref{fig:attestations}; in each time period, a new block is created with the parent block being the last head of chain (which is ideally just the last block created), and validators in the corresponding committee all perfectly attest to the new block. In less idealized situations, the blocks may fork, we may have missing attestations, etc. and the attestations will come to play by helping validators decide the correct head of the chain. In Section~\ref{sec:stage-2} we will define all the details.

\subsection{Byzantine Validators, PBFT}

In a protocol, we call a validator \emph{honest} if he/she follows the protocol, and \emph{byzantine} otherwise. We typically assume strictly less than $p = \frac{1}{3}$ of validators are byzantine. This constant can be traced to Practical Byzantine Fault Tolerance (PBFT) from \cite{castro1999practical}, a classic consensus protocol in byzantine tolerance literature; PBFT ensures that the system runs correctly as long as less than $\frac{1}{3}$ of the \emph{replicas} (synonymous with our \emph{validators}) are byzantine, as proven in \cite{bracha1985asynchronous}. This constant $\frac{1}{3}$ tolerance appears frequently in many works based on PBFT (the main idea is that a pigeonhole principle argument needs the $1/3$ to work, such as in our proof in Section~\ref{sec:safety}), with Casper FFG \cite{buterin2017casper} being the most directly relevant to our work, which we review in Section~\ref{sec:ingredients-ffg}.

Concretely, we define a blockchain running the protocol to be \emph{$p$-slashable} if validators with a total of $pN$ stake can be provably slashed by a validator with the network view, and most of our results will take the form ``at any time, our blockchain either has (good) property X or is $(1/3)$-slashable ,'' because we really cannot guarantee having any good properties when we have many byzantine actors, so this is the strongest type of statement we can have. Also note that having $(1/3)$-slashable is a stronger and much more desirable result than just having $(1/3)$ of the stake belonging to byzantine validators, because if we only had the latter guarantee, we may still incentivize otherwise honest validators to cheat in a way they do not get caught.

One of the most important aspects of PBFT is that it works under \emph{asynchronous} conditions, which means we have no bounds on how long it may take for messages to be received. Thus, it is robust to the demands of blockchains and cryptocurrencies, where we often worry about adversarial contexts. We will discuss synchrony assumptions further in Section~\ref{sec:time-synchrony}.

\subsection{Safety and Liveness}

The ultimate goal for the honest validators is to grow a finalized chain where all blocks form ``logically consistent'' state transitions with each other (despite having validators potentially go offline, suffer latency problems, or maliciously propose conflicting state changes). Formally, this translates to two desired properties: 
\begin{define} \label{def:soundness}
We say a consensus protocol has:
\begin{itemize}
    \item \emph{safety}, if the set of finalized blocks $F(G)$ for any view $G$ can never contain two conflicting blocks. A consequence of having safety is that any validator view $G$'s finalized blocks $F(G)$ can be ``completed'' into a unique subchain of $F(\view(\network))$ that starts at the genesis block and ends at the last finalized block, which we call the \emph{finalized chain}.
    \item \emph{liveness}, if the set of finalized blocks can actually grow. There are different ways to define liveness. For our paper, we say the protocol has:
      \begin{itemize}
        \item \emph{plausible liveness}, if regardless of any previous events (attacks, latency, etc.) it is always possible for new blocks to be finalized (alternatively, it is impossible to become ``deadlocked''). This is to prevent situations where honest validators cannot continue unless someone forfeits their own stake. 
        \item \emph{probabilistic liveness}, if regardless of any previous events, it is probable for new blocks to be finalized (once we make some probabilistic assumptions about the network latency, capabilities of attackers, etc.).
        \end{itemize}
    On first glance, the latter implies the former, but the situation is more subtle; the former is a purely non-probabilistic property about the logic of the protocol; the latter requires (potentially very strong) assumptions about the context of the implementation to guarantee that the protocol ``usually works as intended.'' We discuss this contrast further in Section~\ref{sec:prob-just-to-prob-final}.
\end{itemize}
\end{define}

Our main goal in this paper is to prove these vital properties for our main protocol in Section~\ref{sec:stage-2}.

\begin{remark}[Syntax versus Semantics]

It should not be immediately obvious that chains have anything to do with logical consistence of state transitions, and this requires more work on the part of the protocol designer to ensure. For example, one can dictate a  protocol where a block that allows a user to spend a coin $S$ forbids an ancestor block from using $S$ in another transaction. 
In this situation, if Xander obtained a coin $S$ from block $B_x$ and writes two blocks $B_y$ and $B_z$, such that $B_y$'s data includes Xander paying $S$ to Yeezus  and  $B_z$'s data includes Xander paying $S$ to Zachariah, then the logical idea that these actions are inconsistent corresponds to the graph-theoretical property that $B_y$ and $B_z$ conflict as blocks. The language of chains and conflicting blocks is enough to serve any such logic, so we assume in our paper that the syntax of what blocks can be children of what other blocks has already been designed to embed the semantics of logical consistency, which allows us to ignore logic and only think in terms of chains and conflicting blocks.
\end{remark}

\subsection{Time, Epochs, and Synchrony}
\label{sec:time-synchrony}

In our model, we measure time in \emph{slots}, defined as some constant number of seconds (tentatively $12$ seconds in \cite{beacon}). We then define an \emph{epoch} to be some constant number $C$ (tentatively $64$ slots in \cite{beacon}) of slots. The \emph{epoch $j$ of slot $i$} is $\epoch(i) = j=\lfloor \frac{i}{C} \rfloor$. In other words, the blocks belonging to epoch $j$ have slot numbers $jC+k$ as $k$ runs through $\{0, 1, \ldots, C-1\}$. The genesis block $\gblock$ has slot number $0$ and is the first block of epoch $0$. 

The main purpose of epochs is to divide time into pieces, the boundaries between which can be thought of as ``checkpoints.'' This allows concepts from Casper FFG to be used, as we will see in Section~\ref{sec:ingredients-ffg}.

We should not assume the network is guaranteed to have the validators see the same messages at any given time, or even that different validators have the same view of time. The study of consensus protocols address this issue by having different \emph{synchrony conditions}, such as:
\begin{itemize}
    \item a \emph{synchronous} system has explicit upper bounds for time needed to send messages between nodes;
    \item an \emph{asynchronous} system has no guarantees; recall that PBFT \cite{castro1999practical} works under asynchrony.
    \item a \emph{partially synchronous} system can mean one of two things depending on context: (i) explicit upper bounds for delays exist but are not known a priori; (ii) explicit upper bounds are known to exist after a certain unknown time $T$. The work \cite{dwork1988consensus} establishes bounds on fault tolerance in different fault and synchrony models, focusing on partial synchrony. 
\end{itemize}

For us, we make no synchrony assumptions when studying safety and plausible liveness. When studying probabilistic liveness (i.e. trying to quantify liveness under ``realistic'' conditions), we will use the notion (ii) of partial synchrony above. We say that the network is \emph{$t$-synchronous} at time $T$ (where $T$ and $t$ are both in units of slots) if all messages with timestamps at or before time $(T-t)$ are in the views of all validators at time $T$ and afterwards;  e.g., if each slot is $12$ seconds, then $(1/2)$-synchrony means that all messages are received up to $6$ seconds later. 

\begin{remark}It is possible to e.g., receive messages ``from the future.'' Suppose Alexis sends a message timestamped at 00:01:30 PT and Bob receives it 1 second later with a clock that is 3 seconds behind Alexis's. Bob sees this 00:01:30 PT message on his own clock timestamped at 00:01:28 PT because he is a net total of 2 seconds behind. To make analysis easier, we can assume that all (honest) validators always delay the receiving of a message (i.e., do not add it into their view) until their own timestamps hit the message's timestamp. This allows us to assume no messages are read (again, by honest validators) in earlier slots than they should.
\end{remark}

\section{Main Ingredients}
\label{sec:stage-1}



\subsection{Casper FFG}
\label{sec:ingredients-ffg}

Buterin and Griffith introduce Casper the Friendly Finality Gadget (FFG) in~\cite{buterin2017casper}. This tool defines the concepts of \emph{justification} and \emph{finalization} inspired by practical Byzantine Fault Tolerance (PBFT) literature. Casper is designed to work with a wide class of blockchain protocols with tree-like structures.

\begin{itemize}
    \item Every block has a \emph{height} defined by their distance from the genesis block (which has height $0$). Equivalently, the height of $B$ is the length of $\chain(B) - 1 .$
    \item We define \emph{checkpoint} blocks to be blocks whose height is a multiple of a constant $H$ (in \cite{buterin2017casper}, $H = 100$). We define the \emph{checkpoint height} $h(B)$ for a checkpoint block $B$ as the height of $B$ divided by $H$, which is always an integer. Thus, we can think of the subset of checkpoints in the view as a subtree, containing only blocks whose heights are multiples of $H$.
    \item \emph{Attestations} are signed messages containing ``checkpoint edges'' $A \rightarrow B$, where $A$ and $B$ are checkpoint blocks. We can think of each such attestation as a ``vote'' to move from block $A$ to $B$. The choices of $A$ and $B$ depend on the underlying blockchain and is not a part of Casper. Each attestation has a \emph{weight}, which is the stake of the validator writing the attestation. Ideally, $h(B) = h(A) + 1$, but this is not a requirement. For example, if $H = 100$, it is possible for an honest validator to somehow miss block $200$, in which case their underlying blockchain may want them to send an attestation from checkpoint block $100$ to checkpoint block $300$.
    \end{itemize}

\begin{remark}
Observe we setup time in units of \emph{slots}, which are grouped into \emph{epochs}. This is analogous to (but not exactly identical to) block heights and checkpoint blocks. We discuss this further in Section~\ref{sec:ebb}.
\end{remark}

Casper FFG also introduces the concepts of \emph{justification} and \emph{finalization}, which are  analogous to phase-based concepts in the PBFT literature such as \emph{prepare} and \emph{commit}, e.g., see \cite{castro1999practical}:
    \begin{itemize}
    \item In each view $G$, there is a set of \emph{justified} checkpoint blocks $J(G)$ and a subset $F(G) \subset J(G)$ of \emph{finalized} checkpoint blocks. The genesis block is always both justified and finalized.
    \item In a view $G$, a checkpoint block $B$ is \emph{justified} (by a justified checkpoint block $A$) if $A$ is justified and there are attestations voting for $A \rightarrow B$ with total weight at least $2/3$ of total validator stake. Equivalently, we say there is a \emph{supermajority link} $A \supervotes B$. This is a view-dependent condition, because the view in question may have or have not seen all the relevant attestations to break the $2/3$ threshold (or even having seen the block $B$ itself), which is why the set of justified blocks is parametrized by $G$.
    \item In a view $G$, if  $A \in J(G)$ (equivalently, $A$ is justified), and $A \supervotes B$ is a supermajority link with $h(B) = h(A) + 1$, then we say $A \in F(G)$ (equivalently, $A$ is \emph{finalized}). 
    
\end{itemize}

Finally, Casper introduces \emph{slashing conditions}, which are assumptions about honest validators. When they are broken by a validator $V$, a different validator $W$ can \emph{slash} $V$ (destroy $V$'s stake and possibly getting some sort of ``slashing reward'') by offering proof $V$ violated the conditions. 
\begin{define}
\label{def:casper-slashing} 
The following \emph{slashing conditions}, when broken, cause a validator violating them to have their stake slashed:
\begin{itemize}
    \item (S1) No validator makes two distinct attestations $\alpha_1$ and $\alpha_2$ corresponding to checkpoint edges $s_1 \rightarrow t_1$ and $s_2 \rightarrow t_2$ respectively  with $h(t_1) = h(t_2)$.
    \item (S2) No validator makes two distinct attestations $\alpha_1$ and $\alpha_2$ corresponding respectively to checkpoint edges $s_1 \rightarrow t_1$ and $s_2 \rightarrow t_2$, such that 
    \[
h(s_1) < h(s_2) < h(t_2) < h(t_1).
    \]
\end{itemize}
\end{define}

As Casper is a finality gadget and not a complete protocol, it assumes the underlying protocol has its own fork-choice rule, and at every epoch all the validators run the fork-choice rule at some point to make one (and only one) attestation. It is assumed that honest validators following the underlying protocol will never be slashed. For example, if the protocol asks to make exactly one attestation per epoch, then honest validators will never violate (S1).

The main theorems in Casper are (slightly paraphrased):

\begin{theorem}[Accountable Safety]
Two checkpoints on different branches cannot both be finalized, unless a set of validators owning stake above some total provably violated the protocol (and thus can be held accountable).
\end{theorem}

\begin{theorem}[Plausible Liveness]
It is always possible for new checkpoints to become finalized, provided that new blocks can be created by the underlying blockchain.
\end{theorem}

\begin{remark}
Palmskog et al. ~\cite{palmskog2018verification} provided a mechanical proof assistant for Casper FFG's accountable safety and plausible liveness in the Coq Proof Assistant. The authors modified the blockchain model in Coq, Toychain, to use for the Casper FFG with Ethereum's beacon chain specs in order to see how Casper FFG's proofs will work in practice.
\end{remark}


\subsection{LMD GHOST Fork-Choice rule}
\label{sec:ingredients-lmd}

The Greediest Heaviest Observed SubTree rule (GHOST) is a fork-choice rule introduced by Sompolinsky and Zohar \cite{sompolinsky2015secure}. Intuitively, GHOST is a greedy algorithm that grows the blockchain on sub-branches with the ``most activity.'' Zamfir \cite{lmdghost}, in looking for a ``correct-by-construction'' consensus protocol, introduces a natural adapation of GHOST that happens to be well-suited for our setup\footnote{This part of the field can create phantasmal confusion: the name of  \cite{lmdghost} is ``Casper the Friendly GHOST,'' which introduces Casper CBC (also see \cite{caspercbc}), a completely separate protocol from Casper FFG. Our paper is in some sense recombining these two ideas, but the ``Casper'' in \cite{lmdghost} is very different from Casper FFG.}. We call this variant \emph{Latest Message Driven Greediest Heaviest Observed SubTree (LMD GHOST)}, which we can define after having a notion of \emph{weight}:

\begin{define} \label{def:weight}
Given a view $G$, Let $M$ be the set of latest attestations, one per validator. The \emph{weight} $w(G, B, M)$ is defined to be the sum of the stake of the validators $i$ whose last attestation in $M$ is to $B$ or descendants of $B$.
\end{define}



\begin{algorithm}
    \caption{LMD GHOST Fork Choice Rule.}\label{alg:lmd-ghost}
    \begin{algorithmic}[1]
        \Procedure{LMD-GHOST}{$G$}
        \State $B \gets \gblock$
        \State $M \gets$~the most recent attestations of the validators (one per validator)
        \While{$B$ is not a leaf block in $G$}
            \State $B \gets \displaystyle\argmax_{B' \text{ child of } B} w(G, B', M)$ 
            \State (ties are broken by hash of the block header)
        \EndWhile
        \State \Return{B}
    \EndProcedure
    \end{algorithmic}
\end{algorithm}

\begin{figure}[H]
\begin{center}
\begin{tikzpicture}
\tikzset {
    basic/.style = {rectangle, minimum width=2.2em, minimum height = 1.2em},
    root/.style = {basic, thin, align=center},
    block/.style = {draw, basic, thin, align=center},
    attest/.style = {draw, circle, minimum size=0.15cm, align=center}
}
\node[draw, basic,fill = blue!50](R) {{$\gblock$}};
\node[block, below of = R, fill = blue!50](A) {{$8$}};
\node[block, below of = A, fill = blue!50](B){$8$};
\node[block, below left of = B, fill = blue!50](C){$5$};
\node[block, below right of = B](D){$3$};
\node[block, below right of = C](F){$1$};
\node[block, left of = F](E){$1$};
\node[block, left of = E,fill = blue!50](G){$3$};
\node[block, below of = E](H){$1$};
\node[block, below right of = D](I){$3$};
\node[block, below left of = I](J){$1$};
\node[block, below right of = I](K){$2$};
\node[block, below of = K](L){$2$};
\node[block, below of = L](M){$1$};
\node[attest, right of = L](a7){};
\node[attest, right of = M](a8){};
\node[attest, below of = J](a6){};
\node[attest, left of = G](a1){};
\node[attest, below  left of = G](a2){};
\node[attest, below of = G](a3){};
\node[attest, below of = H](a4){};
\node[attest, below of = F](a5){};
\draw[->] (A) -- (R);
\draw[->] (B) -- (A);
\draw[->] (C) -- (B);
\draw[->] (D) -- (B);
\draw[->] (E) -- (C);
\draw[->] (F) -- (C);
\draw[->] (G) -- (C);
\draw[->] (H) -- (E);
\draw[->] (I) -- (D);
\draw[->] (J) -- (I);
\draw[->] (K) -- (I);
\draw[->] (L) -- (K);
\draw[->] (M) -- (L);
\draw[->] (a7) -- (L);
\draw[->] (a8) -- (M);
\draw[->] (a6) -- (J);
\draw[->] (a1) -- (G);
\draw[->] (a2) -- (G);
\draw[->] (a3) -- (G);
\draw[->] (a4) -- (H);
\draw[->] (a5) -- (F);

\end{tikzpicture}
\end{center}
\caption{\label{fig:lmd-ghost} An example of the LMD-GHOST fork-choice rule. The number in each block $B$ is the weight (by stake), with all attestations (circles) having weight $1$ in our example. A validator using this view will conclude the blue chain to be the canonical chain, and output the latest blue block on the left, with weight $3$, to be the head of the chain.}
\end{figure}
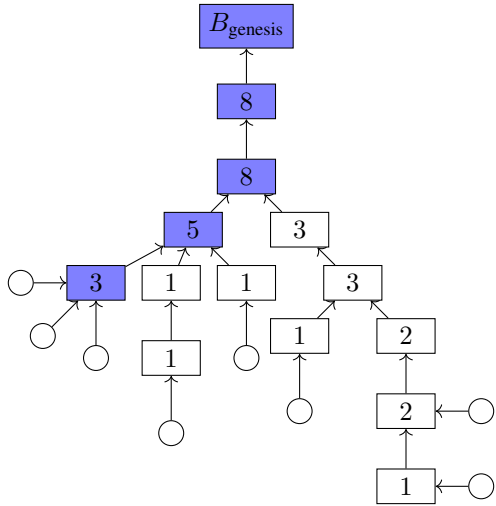

The idea of LMD GHOST is that at any fork, we use the weights of the subtrees created by the fork as a heuristic and assume the subtree with the heaviest weight is the ``right'' one, as evident from the name of the algorithm. We will always end up at a leaf block, which defines a canonical chain. See Figure~\ref{fig:lmd-ghost} for a graphical visualization. In \cite{lmdghost}, LMD GHOST was tied to a particular protocol. Our presentation of it in this paper treats the specifics of attestations as details to be filled in by the protocol, like how Casper does not specify details on how blocks are constructed. We merely need that the concepts of \emph{blocks} and \emph{attestations} exist, and then LMD GHOST defines an algorithm to find a block given a view. In Section~\ref{sec:stage-2} we will modify the rule to satisfy our other purposes.




\section{Main Protocol: Gasper}
\label{sec:stage-2}

We now define our main protocol $\mainprotocol$, which is a combination of the GHOST and Casper FFG ideas. The main concepts are:
\begin{itemize}
    \item \emph{(Epoch boundary) pairs} of a chain: given a chain, certain blocks are picked out, ideally one per epoch, to play the role of Casper's  \emph{checkpoints}. However, a block may appear more than once as a checkpoint on the same chain (this is a nuance not found in Casper; we expound on this in Section~\ref{sec:ebb}), so we use ordered pairs $(B,j)$, where $B$ is a block and $j$ is an epoch, to disambiguate. These will be called \emph{epoch boundary pairs}, or \emph{pairs} for short.
    
    \item \emph{Committees}: the validators are partitioned into \emph{committees} in each epoch, with one committee per slot. In each slot, one validator from the designated committee proposes a block. Then, all the members of that committee will attest to what they see as the head of the chain (which is hopefully the block just proposed) with the fork-choice rule HLMD GHOST (a slight variation of LMD GHOST).

    \item \emph{Justification and Finalization}: these concepts are virtually identical to that of Casper FFG, except we justify and finalize pairs instead of justifying and finalizing checkpoint blocks, i.e., given a view $G$, $J(G)$ and $F(G)$ are sets of pairs. 
\end{itemize}

\subsection{Epoch Boundary Blocks and Pairs}
\label{sec:ebb}

Recall any particular block $B$ uniquely determines a chain, $\chain(B)$.  For a block $B$ and an epoch $j$, define $\ebb(B, j)$, the $j$-th \emph{epoch boundary block} of $B$, to be the block with the highest slot less than or equal to $jC$ in $\chain(B)$. Let the latest such block be $\lastepoch(B)$, or the \emph{last epoch boundary block} (of $B$). We make a few observations:
\begin{itemize}
\item For every block $B$, $\ebb(B, 0) = \gblock.$ 
\item More generally, if $\slot(B) = jC$ for some epoch $j$, $B$ will be an epoch boundary block in every chain that includes it.
\item However, without such assumptions, a block could be an epoch boundary block in some chains but not others.
\end{itemize}

To disambiguate situations like these, we add precision by introducing \emph{epoch boundary pairs} (or \emph{pairs} for short) $(B,j)$, where $B$ is a block and $j$ is an epoch. Our main concepts of justification and finalization will be done on these pairs. Given a pair $P = (B,j)$, we say $P$ has  \emph{attestation epoch} $j$, using the notation $\attestepoch(P) = j$, which is \textbf{not} necessarily the same as $\epoch(B)$.


\begin{figure}[H]
\begin{center}
\begin{tikzpicture}[edge from parent/.style={draw,latex-}, level distance=25mm,sibling distance=10mm, grow=east]


\node [draw] at (0, 0) (63) {63};
\node [draw] at (2, 0) (64) {64};
\node [draw, dashed] at (2, -2) (63') {``63''};

\node [draw] at (4, 0) (65) {65};
\node [draw] at (6, -2) (66) {66};

\draw [->] (65) edge (64) (64) edge (63);
\draw [->] (66) edge (63);

\draw [->, dashed] (66) edge (63') (63') edge (63);





\draw (1.4,0.8) -- (1.4,-2.8);

\node[text width=2cm, font=\fontsize{8}{0}\selectfont]  at (1,1) {Epoch 0};
\node[text width=2cm, font=\fontsize{8}{0}\selectfont]  at (5,1) {Epoch 1};


\node[text width=3cm, font=\fontsize{8}{0}\selectfont]  at (8,0) {$\lastepoch(65) = 64$};
\node[text width=3cm, font=\fontsize{8}{0}\selectfont]  at (8,-2) {$\lastepoch(66) = 63$};

\node[text width=3cm, font=\fontsize{8}{0}\selectfont]  at (2.9,-2.8) {$\ebb(66,1) = (63, 1)$};
\end{tikzpicture}
\end{center}

\caption{\label{fig:ebb-example} Example of epoch boundary blocks and pairs. Blocks are labeled with their slot numbers. ``$63$'' is not an actual block but illustrates the perspective of $66$ needing an epoch boundary block at slot $64$ and failing to find one.}
\end{figure}
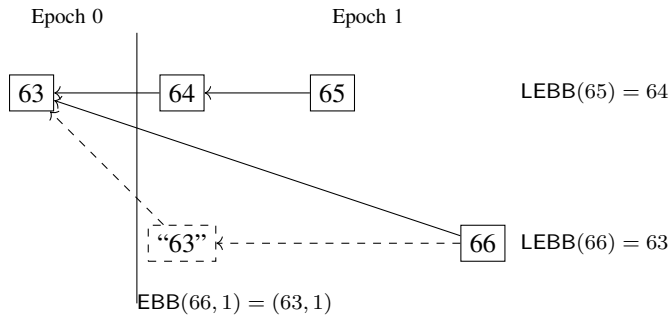

\begin{ex}
Suppose blocks (labeled by slots) form a chain $63 \leftarrow 64 \leftarrow 65$ and there is a fork $63 \leftarrow 66$. Then $\ebb(65, 1) = 64$ since $64$ is in $\chain(65)$. When we try to find $\ebb(66,1)$, we look for block $64$ because we are in epoch $1$ and realize it is not in $\chain(66)$, so we have to look backwards and ``pull up'' block $63$ from the previous epoch to serve as the epoch boundary block in this chain. This is visualized by the dashed components in Figure~\ref{fig:ebb-example}. Observe even though $\attestepoch(63,1)$ and $\epoch(63) = 0$, $\epoch()$ is a local property that only depends on the block's slot, and epoch boundary concepts like $\attestepoch()$ depend on the context of the chain.
\end{ex}

\begin{remark}
We further address why we use pairs instead of checkpoint blocks. Casper FFG is a ``finality gadget,'' meaning it is designed to place a layer of finality on top of a blockchain which has probabilistic liveness, which gives a steady new supply of checkpoint blocks. Probabilistic liveness assumes good synchrony conditions. For the safety part of the design, we would want to make fewer assumptions and take into account worst-case scenarios where we may be put into a state where we have not seen a new block for a while. 

As an example, block $B$ may be the head of the chain from epoch $1$, but be currently at epoch $3$ with no new block built on top of $B$. In the original Casper FFG, we expect probabilistic liveness so we would have a different checkpoint at every epoch. However, when we have no synchrony assumptions, in the analysis we would need to differentiate the idea of ``checkpoint in epoch 2'' and ``checkpoint in epoch 3'' even though the best candidate for both is $B$, which is itself still only in epoch $1$! This naturally induces the concepts $(B,2)$ and $(B,3)$, which represents the best approximation to a checkpoint in a later epoch if the checkpoint block that ``should have been there'' is missing.

Also, Casper FFG makes no assumptions about time in the underlying blockchain -- only block heights are important, with no notion of \textit{slots} or \textit{epochs}. The use of block height is a natural choice in proof-of-work blockchains due to the Poisson process of mining new blocks serving somewhat as a system clock, with blocks coming at fairly irregular times. $\mainprotocol$, as a proof-of-stake protocol, can have blocks coming in at controlled regular intervals as part of the protocol (instead of depending on a random process), so instead the notion of time is explicitly desired in the protocol. To capture this notion of time, each object in our blockchain then naturally requires the knowledge of both the data (captured in a block) and the time (captured in the epoch count), which also leads naturally to the idea of pairs.


\end{remark}


\subsection{Committees}

The point of committees is to split up the responsibilities among the validators. To start, assume the validators have access to a sequence of random length-$N$ permutations $\rho_0, \rho_1, \ldots $, as functions $\{1, 2, \ldots, N\} \rightarrow \{1, 2, \ldots, N\}$. In the scope of this article, we assume we get these random permutations from a random oracle\footnote{Achieving non-exploitable randomness on the blockchain is itself an interesting problem, stimulating research such as verifiable delay functions; however, in this article we will take this randomness for granted.}. 

Recall time is split into epochs, of $C$ slots each. Each permutation $\rho_j$ will be used only during epoch $j$. Its role is to pseudorandomly select validators into $C$ \emph{committees}, each of whom has responsibilities for one slot of the epoch. To be precise:
\begin{itemize}

\item During epoch $j$, we would like to split the set of validators $V$ into $C$ equal-size committees $S_0, S_1, \ldots, S_{C-1}$ (we assume $C|N$ for easier notation; dealing with ``roughly equal" size committees does not change the essence of our approach).
\item Therefore, for each $k \in \{0,1,  \ldots, C-1\}$, we define $S_k$ to be the set of the $N/C$ validators of the form $\validatorset_{\rho_j(s)}$, where $s \equiv k \pmod{C}$. Note, for an epoch $j$, the sets $S_0, S_1, \ldots, S_{C-1}$ partition the entire set of validators $\{V_1, \ldots, V_N  \}$ across all the slots of epoch $j$, as desired.
\end{itemize}

To summarize, in each epoch $j$, $\rho_j$ allows us to shuffle the validators into $C$ committees. Our work does not assume more than this intuition.

\subsection{Blocks and Attestations}

Now, in each slot, the protocol dictates $2$ types of ``committee work'' for the committee assigned to their respective slot: one person in the committee needs to \emph{propose} a new block, and everyone in the committee needs to \emph{attest} to their head of the chain. Both proposing blocks and attesting mean immediately adding a corresponding message (a block and an attestation, respectively) to the validator's own view and then broadcasting the message to the network. Recall the messages proposing (non-genesis) blocks and publishing attestations have digital signatures. 

Both proposing and attesting requires the committee member to run the same fork-choice rule $\mlmd()$ on the validator's own view, which is a variation of the LMD GHOST fork-choice rule. $\mlmd()$'s definition requires concepts that we have not introduced yet, so we delay its definition to Section~\ref{sec:hlmd-ghost}. For now, the only thing we need to know is, like LMD GHOST, $\mlmd()$ takes a view $G$ and returns a leaf block $B$ as the head of the chain, making $\chain(B)$ the canonical chain of $G$.

We now present the responsibilities of the protocol during slot $i = jC+k$, where $k \in \{0, 1, \ldots, C-1\}$. All mentions of time are computed from the point-of-view of the validator's local clock, which we assume to be synced within some delta.

\begin{define}
At the beginning of slot $i = jC+k$, designate validator $V = V_{\rho_j(k)}$, i.e., the first member of the committee $S_k$ of epoch $j$, as the \emph{proposer} for that slot. The proposer computes the canonical head of the chain in his/her view, i.e., $\mlmd(\view(V, i)) = B'$, then \emph{proposes a block} $B$, which is a message containing:
    \begin{enumerate}
    \item $\slot(B) = i$, the slot number.
    \item $P(B) = B'$, a pointer to the parent block; in other words, we always build a block on top of the head of the chain.
    \item $\newattests(B)$, a set of pointers to all the \emph{attestations} (to be defined next) $V$ has accepted, but have not been included in any $\newattests(B')$ for a block $B'$ who is an ancestor of $B$. 
    \item Some implementation-specific data (for example, ``Yunice paid 4.2 ETH to Brad'' if we are tracking coins), the semantics of which is irrelevant for us.
    \end{enumerate}
For dependencies, $B$ depends on $P(B)$ and all attestations in $\newattests(B)$. (so for example, if we see a block on the network but do not see its parent, we ignore the block until we see its parent). 
\end{define}

In typical (honest) behavior, we can assume each slot number that has a block associated with it also has at most one such block in $\view(\network)$. By default, $\gblock$ is the unique block with slot $0$. A dishonest validator can theoretically create a block with a duplicate slot number as an existing block, but we can suppose the digital signatures and the pseudorandom generator for block proposal selection are set up so such behavior would be verifiably caught. For example, every time a digitally signed block enters the view of a validator, the validator can immediately check if the block proposer is the unique person allowed to propose a block for that respective slot. A validator can also prove the same person proposed two blocks in the same slot by just pointing to both blocks. 

\begin{define} 
At time $(i+1/2)$, the middle of slot $i = jC+k$, each validator $V$ in committee $S_k$ computes $B' = \mlmd\left(\view(V, i+1/2)\right)$, and publishes an \emph{attestation} $\alpha$, which is a message containing:
    \begin{enumerate}
        \item $\slot(\alpha) = jC+k$, the slot in which the validator is making the attestation. We will also use $\epoch(\alpha)$ as shorthand for $\epoch\left(\slot(\alpha)\right)$.
        \item $\block(\alpha) = B'$. We say $\alpha$ \emph{attests to $\block(\alpha)$}. We will have $\slot\left(\block(\alpha)\right) \leq \slot(\alpha)$, and  ``usually'' get equality by quickly attesting to the block who was just proposed in the slot. 
        \item A \emph{checkpoint edge} $\lastjustpair(\alpha) \votes  \lastepochpair(\alpha)$. Here, $\lastjustpair(\alpha)$ and $\lastepochpair(\alpha)$ are epoch boundary pairs in $\view(V, i+1/2)$. We define them properly in Section~\ref{subsec:justification}. 
    \end{enumerate}
    For dependencies, $\alpha$ depends on $\block(\alpha)$. So we ignore an attestation until the block it is attesting to is accepted into our view (this is one of the bigger differences between theory and implementation; we discuss this more in in Section~\ref{sec:practice-vs-theory}).
\end{define}

Intuitively, $\alpha$ is doing two things at once: it is simultaneously a ``GHOST vote'' for its block and also a ``Casper FFG vote'' for the transition between the two epoch boundary pairs (akin to Casper's checkpoint blocks).



\subsection{Justification}
\label{subsec:justification}

\begin{define}
\label{def:dependency-view}
Given a block $B$, we define\footnote{We have tyrannically overworked the notation $\view()$ by this point, but there should be no ambiguity when we know the type of its parameter.} $\view(B)$, the \emph{view of $B$}, to be the view consisting of $B$ and all its ancestors in the dependency graph. We define $\ffgview(B)$, the \emph{FFG view of $B$}, to be $\view(\lastepoch(B))$.
\end{define}

The definition $\view(B)$ is ``agnostic of the viewer'' in which any view that accepted $B$ can compute an identical $\view(B)$, so we do not need to supply a validator (or $\network$) into the argument. Intuitively, $\view(B)$ ``focuses'' the view to $\chain(B)$ and $\ffgview(B)$ looks at a ``frozen'' snapshot of $\view(B)$ at the last checkpoint. Casper FFG operates only on epoch boundary pairs, so the FFG view of a block $B$ extracts exactly the information in $\chain(B)$ that is relevant to Casper FFG.

We now define the main concepts of justification. To start, recall an attestation $\alpha$ for $\mainprotocol$ contains a \emph{checkpoint edge} $\lastjustpair(\alpha) \votes \lastepochpair(\alpha)$, acting as a ``FFG vote'' between two epoch boundary pairs. After a couple of new definitions, we can finally explicitly define these notions.


\begin{define} \label{def:supermajority}
We say there is a \emph{supermajority link} from pair $(A, j')$ to pair $(B, j)$ if the attestations with checkpoint edge $(A,j') \votes (B, j)$ have total weight more than $\frac{2}{3}$ of the total validator stake. In this case, we write $(A, j') \supervotes (B,j)$. 
\end{define}

\begin{define} \label{def:justify}
Given a view $G$, we define the set $J(G)$ of \emph{justified} pairs as follows: 
\begin{itemize}
    \item $(\gblock, 0) \in J(G)$;
    \item if $(A,j') \in J(G)$ and $(A,j') \supervotes (B,j)$, then $(B,j) \in J(G)$ as well.
\end{itemize}
If $(B,j) \in J(G)$, we say \emph{$B$ is justified in $G$ during epoch $j$}.
\end{define}

\begin{define} \label{def:checkpoint-edge}
Given an attestation $\alpha$, let $B = \lastepoch(\block(\alpha))$. We define:
\begin{enumerate}
\item $\lastjustpair(\alpha)$, the \emph{last justified pair of $\alpha$}: the highest attestation epoch (or last) justified pair in $\ffgview(\block(\alpha)) = \view(B)$. 
\item $\lastepochpair(\alpha)$, the \emph{last epoch boundary pair of $\alpha$}: $(B, \epoch(\slot(\alpha)))$.
\end{enumerate}
\end{define}



While we define $J(G)$ for any view $G$, we are typically only interested in $J(G)$ for $G = \ffgview(B)$. We can think of each chain as having its own ``state'' of justified blocks/pairs that is updated only at the epoch boundaries.


\begin{figure}[H]
\begin{center}
\begin{tikzpicture}[edge from parent/.style={draw,-latex}, level distance=10mm,sibling distance=10mm, grow=right, scale=0.90]
\node[minimum height=25, fill=purple!50, draw]  at (-0.2,0) {0}
child {node [draw] {1}
		child {node {\dots}
			child {node [draw] {63}
				child {node [draw, fill=purple!50, minimum height=25] {64}}}}};
				


\draw[] (4.1,0) -- ++(1.9,0);
\draw[] (10.8,0) -- ++(0.8,0);

\node [fill=purple!50, minimum height=25, draw] at (6.4,0){64}
	child [grow=south east] {node [draw]{130}}
	child [grow=right] {node [draw]{129}
		child {node [draw]{131}
			child {node {\dots}
				child {node [draw] {180}}}}};
				
\node [minimum height=25, draw] at (12,0){180}
	child {node [draw]{193}
    	child [grow=north east] {node [draw, circle]{$\alpha$}}};

\draw[dashed] (10.6,0.7) -- ++(2.7,0);

\draw [->,double] (0,-0.3) arc (250:290:5);
\draw [->,double] (0,0.3) arc (110:70:8.8);
\draw [->] (6.8,0.3) arc (110:70:7.1);

\draw (3.3,1.5) -- (3.3,-1);
\draw (5.9,1.5) -- (5.9,-1);
\draw (11.4,1.5) -- (11.4,-1);

\node[text width=2cm, font=\fontsize{8}{0}\selectfont]  at (2.,1.5) {Epoch 0};
\node[text width=2cm, font=\fontsize{8}{0}\selectfont]  at (5.1,1.5) {Epoch 1};
\node[text width=2cm, font=\fontsize{8}{0}\selectfont]  at (9.2,1.5) {Epoch 2};
\node[text width=2cm, font=\fontsize{8}{0}\selectfont]  at (13.5,1.5) {Epoch 3};


\end{tikzpicture}
\end{center}
    \caption{\label{fig:attestation-ebbs} A validator's view $G$ as she writes an attestation in epoch $3$. During epoch $1$, latency issues make her not see any blocks, so block $64$ is both $\ebb(193,1)$ and $\ebb(193,2)$. She ends up writing an $\alpha$ with a ``GHOST vote'' for $\block(\alpha) = 193$ and a ``FFG vote'' checkpoint edge (single arc edge)  $(64, 2) \votes (180, 3)$ for $\alpha$. We use slot numbers for blocks, so block $0$ is just $\gblock$. Blocks in red are justified (in $G$). Double edges corresponding to supermajority links.}
\end{figure}
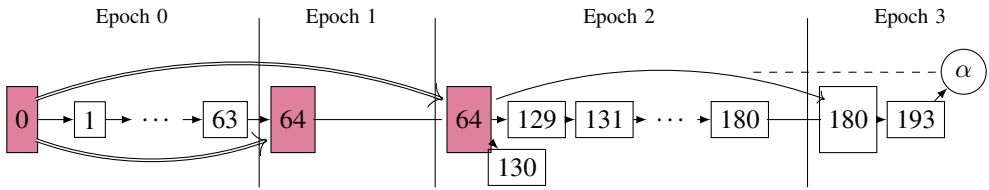

\begin{ex}
A validator runs $\mlmd(G)$ on her view $G$ to obtain her head of the chain, which is block $193$. She is then supposed to attest to $193$ with an attestation $\alpha$.

On $\chain(193)$ portrayed in Figure~\ref{fig:attestation-ebbs}, $\lastepochpair(\alpha) = (180,3)$, even though $\epoch(180) = 2$, because our attestation has epoch $3$, and we were looking for $3*64=192$ but did not see it (one can imagine, like in our figure, that we ``pull up'' block $180$ to show that it is $\ebb(193, 3)$). In $\ffgview(193) = \view(180)$, the last justified (by epoch number, not slot) pair is $(64,2)$, so $\lastjustpair(\alpha) = (64,2)$. 

Note that this relies on $\view(180)$ to include attestations worth at least $(2N/3)$-stake pointing to $(64,2)$. It would be possible to have a case where this did not happen by $180$, but the chain does see $(2N/3)$ worth of attestations by $193$. In this case, $(64,2)$ is in $J(G)$, but $(64,2)$ is not in $J(\ffgview(193))$, so the resulting checkpoint edge would be, e.g., $(64,1) \votes (180,3)$ instead, assuming the attestations for $(64,1)$ are in $\view(180)$. The checkpoint edge could even end up being $(0,0) \votes (180,3)$, if $180$ did not see those attestations; note that this is very possible if the attestations for $(64,1)$ were only included in the blocks that forked off of $\chain(193)$ (such as block $130$), as the forking may be an indication of network problems causing those attestations to be temporarily unavailable to validators on $\chain(193)$.
\end{ex}

Everything here is analogous to Casper FFG. Inside the chain of $\block(\alpha)$ is a sub-chain created by the epoch boundary blocks of that chain, starting from $\gblock$ and ending at $B = \lastepoch(\alpha)$. We want to focus on this subchain of blocks (represented by pairs to allow for boundary cases) and justify epoch boundary pairs with many attestations; this sub-chain is exactly what's captured by $\ffgview(\block(\alpha)))$. $\alpha$ is a vote to transition from some last justified pair $(B', j')$ to the new pair $(B, j)$, visualized as the checkpoint edge $(A, j') \votes (B,j)$. If $(2N/3)$ stake worth of such votes happen, we create a supermajority link $(A, j') \supervotes (B,j)$ and justify $(B,j)$. 

\subsection{Finalization}
Given our notion of \emph{justification} and a new fork-choice rule, we are now ready to define the notion of \emph{finalization}. Finalization is a stronger notion of justification in the sense that the moment \textbf{any} view considers a block $B$ as finalized for some $j$, \textbf{no} view will finalize a block conflicting with $B$ unless the blockchain is $(1/3)$-slashable. 


\begin{define}
\label{def:finalization}

For a view $G$, we say $(B_0, j)$ is \emph{finalized}  (specifically, \emph{$k$-finalized}) in $G$ if $(B_0, j) = (\gblock, 0)$ or if there is an integer $k \geq 1$ and blocks $B_1, \ldots, B_k \in \view(G)$ such that the following holds:
\begin{itemize}
\item $(B_0, j), (B_1, j+1), \ldots, (B_k, j+k)$ are adjacent epoch boundary pairs in $\chain(B_k)$;
\item $(B_0, j), (B_1, j+1), \ldots, (B_{k-1}, j+k-1)$ are all in $J(G)$;
\item $(B_0, j) \supervotes (B_k, j+k)$.
\end{itemize}

We define $F(G)$ to be the set of finalized pairs in the view $G$; we also say that a block $B$ is \emph{finalized} if $(B,j) \in F(G)$ for some epoch $j$.
\end{define}

We expect $1$-finalized blocks for the vast majority of time. Specifically, this just means we have blocks $B_0$ and $B_1$ such that $(B_0, j)$ is justified in $G$ and $(B_0, j) \supervotes (B_1, j+1)$, or ``a justified epoch boundary block justifies the next adjacent epoch boundary block.'' This is essentially the original definition of \emph{finalized} in \cite{buterin2017casper}. With $(1/2)$-synchrony and without implementation details that interfere with accepting messages, we should only get $1$-finalization. We include the $k > 1$ cases to account for situations where network latency and attestation inclusion delays (see Section~\ref{sec:attestation_inclusion_delay}) are relevant. 


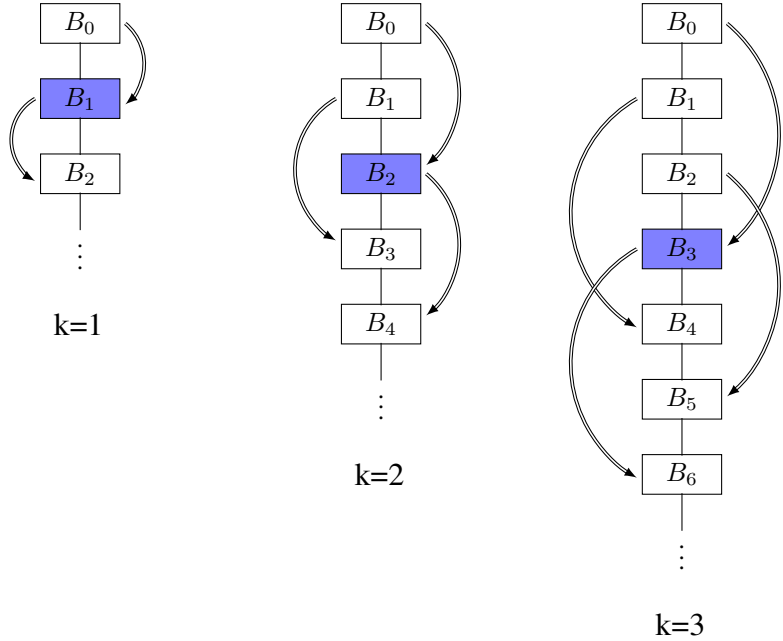
\begin{figure}
\begin{center}

\tikzset {
    basic/.style = {rectangle, minimum width=3em, minimum height = 1.3em},
    root/.style = {basic, thin, align=center},
    block/.style = {draw, basic, thin, align=center}
}
\begin{tikzpicture} [
    level 1/.style={sibling distance = 40mm},
    edge from parent/.style={}, >=latex]

\node [](n) {}
	child {node [block](c1){$B_{0}$}}		
    child {node [block](c2){$B_{0}$}}
    child {node [block](c3){$B_{0}$}};

\begin{scope}[node/.style={block}]
\node [block, below of = c1, , fill = blue!50] (c11) {$B_{1}$};
\node [block, below of = c11] (c12) {$B_{2}$};
\node [below of = c12] (c13) {$\vdots$};
\end{scope}

\begin{scope}[node/.style={block}]
\node [block, below of = c2] (c21) {$B_{1}$};
\node [block, below of = c21, fill = blue!50] (c22) {$B_{2}$};
\node [block, below of = c22] (c23) {$B_{3}$};
\node [block, below of = c23] (c24) {$B_{4}$};
\node [below of = c24] (c25) {$\vdots$};
\end{scope}

\begin{scope}[node/.style={block}]
\node [block, below of = c3] (c31) {$B_{1}$};
\node [block, below of = c31] (c32) {$B_{2}$};
\node [block, below of = c32, fill = blue!50] (c33) {$B_{3}$};
\node [block, below of = c33] (c34) {$B_{4}$};
\node [block, below of = c34] (c35) {$B_{5}$};
\node [block, below of = c35] (c36) {$B_{6}$};
\node [below of = c36] (c37) {$\vdots$};
\end{scope}

\draw [->,double] (-3.4,-1.5) arc (50:-50:20pt); 
\draw [->,double] (-4.6,-2.5) arc (120:240:18 pt); 

\draw [->,double] (0.6,-1.5) arc (50:-50:35pt);  
\draw [->,double] (-0.6,-2.5) arc (120:240:31pt);  
\draw [->,double] (0.6,-3.5) arc (50:-50:35pt);  

\draw [->,double] (4.6,-1.5) arc (50:-50:55pt);  
\draw [->,double] (3.4,-2.5) arc (120:240:50pt);  
\draw [->,double] (4.6,-3.5) arc (50:-50:55pt);  
\draw [->,double] (3.4,-4.5) arc (120:240:50pt);  

\draw [-] (-4,-1.75) -- (-4,-2.25);
\draw [-] (-4,-2.75) -- (-4,-3.25);
\draw [-] (-4,-3.75) -- (-4,-4.25);

\draw [-] (0,-1.75) -- (0,-2.25);
\draw [-] (0,-2.75) -- (0,-3.25);
\draw [-] (0,-3.75) -- (0,-4.25);
\draw [-] (0,-4.75) -- (0,-5.25);
\draw [-] (0,-5.75) -- (0,-6.25);

\draw [-] (4,-1.75) -- (4,-2.25);
\draw [-] (4,-2.75) -- (4,-3.25);
\draw [-] (4,-3.75) -- (4,-4.25);
\draw [-] (4,-4.75) -- (4,-5.25);
\draw [-] (4,-5.75) -- (4,-6.25);
\draw [-] (4,-6.75) -- (4,-7.25);
\draw [-] (4,-7.75) -- (4,-8.25);

\node[text width = 4cm, font=\fontsize{12}{0}\selectfont] at (-2.35, -5.5) {k=1};
\node[text width = 4cm, font=\fontsize{12}{0}\selectfont] at (1.65, -7.5) {k=2};
\node[text width = 4cm, font=\fontsize{12}{0}\selectfont] at (5.65, -9.5) {k=3};

\end{tikzpicture}
\end{center}
\caption{Illustrative examples of Definition~\ref{def:finalization} for $k=1, 2, 3$, respectively. Double arrows are supermajority links. The blue block is the one being finalized in each case.}
\label{fig:finalization-3-cases}
\end{figure}

\begin{remark}
In Figure~\ref{fig:finalization-3-cases}, we see some examples of finalization. On the left, we have $1$-finalization, which is what we expect to happen in the vast majority of the time.  At the center, we have an example of $2$-finalization to account for attestation delays. The $k=3$ example is mainly for the sake of illustration; it takes some contrived orchestration to create. In practice, the planned implementation for finalization in \cite{beacon} does not even include finalization for $k \geq 3$. See Section~\ref{sec:practice-finalization}.


If we define a set of blocks $F$ as finalized and prove safety for them, then safety remains true automatically if we change the definition of \emph{finalized} to any subset of $F$, because safety is defined by the lack of incomparable pairs of finalized blocks. Thus, it does not hurt to define a very general class of blocks as finalized if we can prove safety. This is why we include all $k$ in our definition; the more general definition is easier to analyze.
\end{remark}


\subsection{Hybrid LMD GHOST}
\label{sec:hlmd-ghost}


We now have enough notation to define $\mlmd$. Since the final design is fairly complicated at first glance, we first present Algorithm~\ref{alg:hlmd-ghost-protoype}, a ``prototype'' version; simply put, it starts with the last justified block in the view and then runs LMD GHOST. It is not obvious $(B_J, j)$ should be unique, but it can be proven to be well-defined as a consequence of Lemma~\ref{lem:unique-jep} when the blockchain is not $(1/3)$-slashable. For implementation it is costless to add a tiebreaker (say via block header hash).

\begin{algorithm}[H]
\caption{Prototype Hybrid LMD GHOST Fork Choice Rule}\label{alg:hlmd-ghost-protoype}
\begin{algorithmic}[1]
\Procedure{$\mlmd$}{$G$}
\State $(B_J, j) \gets$~the justified pair with highest attestation epoch $j$ in $G$ 
\State $B \gets B_J$
\State $M \gets$~most recent attestations (one per validator)
\While{$B$ is not a leaf block in $G$}
\State $B \gets \displaystyle\argmax_{B' \text{ child of } B} w(G, B', M)$
\State (ties are broken by hash of the block header)
\EndWhile
\State \Return{B}
\EndProcedure
\end{algorithmic}
\end{algorithm}

However, using this prototype will run into some subtle problems. Because the ``FFG part'' of an attestation is a vote between epoch boundary pairs, and $\lastjustpair(B)$ depends only on $\ffgview(B)$, we can think of the ``last justified pair in a chain'' as ``frozen'' at the beginning of every epoch, which does not mix well with the ``non-frozen'' $(B_J, j)$ being instinctively defined as the last justified block in our entire view, as the latter definition can shift during an epoch (the reader is encouraged to think of problems with this, as an exercise). Thus, we need a version of $(B_J, j)$ that does not change in the middle of an epoch.

Another problem occurs when the algorithm forks: forked blocks can have drastically different last justified pairs, even if they are next to each other. As a result of this mismatch, we can create pathological situations where an honest validator may slash themselves when following Algorithm~\ref{alg:hlmd-ghost-protoype}, if a validator previously attested to a higher last justification epoch but then forks into a chain whose last justification epoch is older (c.f.  Definition~\ref{def:slashing}, S2).

To resolve these issues, first we ``freeze'' the state of the latest justified pair $(B_J, j)$ to the beginning of the epochs; formally, this means when defining $(B_J, j)$ we consider the views of $\ffgview(B_l)$ over the leaf blocks $B_l$ as opposed to the entire view. Then, we filter the branches so we do not go down branches with leaf nodes $B_l$ where $\lastjustpair(B_l)$ has not ``caught up'' to $(B_J, j)$; formally, we create an auxiliary view $G'$ from $G$ that eliminates the problematic branches. These fixes give  Algorithm~\ref{alg:hlmd-ghost}, the final design.

\begin{algorithm}[H]
\caption{Hybrid LMD GHOST Fork Choice Rule}\label{alg:hlmd-ghost}
\begin{algorithmic}[1]
\Procedure{$\mlmd$}{$G$}
\State $L \gets \text{set of leaf blocks $B_l$ in $G$}$
\State \begin{varwidth}[t]{\linewidth}
      $(B_J, j) \gets$~the justified pair with highest attestation epoch $j$ in\par
        \hskip\algorithmicindent
        $J(\ffgview(B_l))$ over $B_l \in L$
      \end{varwidth}
\State $L' \gets \text{set of leaf blocks $B_l$ in $G$ such that  $(B_j, j) \in J(\ffgview(B_l))$} $ 
\State $G' \gets \text{the union of all chains $\chain(B_l)$ over $B_l \in L'$}$
\State $B \gets B_J$
\State $M \gets$~most recent attestations (one per validator)
\While{$B$ is not a leaf block in $G'$}
\State $B \gets \displaystyle\argmax_{B' \text{ child of } B} w(G', B', M)$
\State (ties are broken by hash of the block header)
\EndWhile
\State \Return{B}
\EndProcedure
\end{algorithmic}
\end{algorithm}

A dual (and possibly simpler) way of understanding Algorithm~\ref{alg:hlmd-ghost} is from the perspective of a state-based implementation of  the protocol. We can think of each chain of a leaf block $B_l$ as storing the state of its own last justified pair. During an epoch, new attestations to blocks in the chain updates the GHOST-relevant list of latest attestations $M$ but \textbf{not} the FFG-relevant justification and finalization information of the chain until the next epoch boundary block. This way, the ``FFG part'' of the protocol always works with the ``frozen until next epoch'' information, while the ``GHOST part'' of the protocol is being updated continuously with the attestations. This careful separation allows us to avoid pathological problems from mixing the two protocols.


\subsection{Slashing Conditions}

In this subsection, we add \emph{slashing} conditions, analogous to those from Casper in Definition~\ref{def:casper-slashing}. We prove a couple of desired properties, after which we are ready to use these conditions to prove safety of $\mainprotocol$ in Section~\ref{sec:safety}.

\begin{define} \label{def:slashing}
We define the following \emph{slashing conditions}:
\begin{itemize}
    \item[(S1)] No validator makes two distinct attestations $\alpha_1, \alpha_2$ with $\epoch(\alpha_1) = \epoch(\alpha_2)$. Note this condition is equivalent to $\attestepoch\left(\lastepochpair\left(\alpha_{1}\right)\right) = \attestepoch\left(\lastepochpair\left(\alpha_{2}\right)\right)$.
    \item[(S2)] No validator makes two distinct attestations $\alpha_1, \alpha_2$ with 
    \[
\attestepoch\left(\lastjustpair\left(\alpha_{1}\right)\right)<
\attestepoch\left(\lastjustpair\left(\alpha_{2}\right)\right)<
\attestepoch\left(\lastepochpair\left(\alpha_{2}\right)\right)<
\attestepoch\left(\lastepochpair\left(\alpha_{1}\right)\right).
    \]
\end{itemize}
\end{define}

We now prove a very useful property of the protocol: unless we are in the unlikely situation that we have enough evidence to slash validators with at least $1/3$ of the total stake, in a view $G$ we may assume all elements in $J(G)$ have unique attestation epochs (in other words, the view sees at most one pair justified per epoch).


\begin{lemma}
\label{lem:unique-jep} In a view $G$, for every epoch $j$, there is at most $1$ pair $(B,j)$ in $J(G)$, or the blockchain is $(1/3)$-slashable. In particular, the latter case means there must exist $2$ subsets $\validatorset_1, \validatorset_2$ of $\validatorset$, each with total weight at least $2N/3$, such that their intersection violates slashing condition (S1).
\end{lemma}
\begin{proof}
Suppose we have $2$ distinct pairs $(B, j)$ and $(B', j)$ in $J(G)$. This means in epoch $j$, more than a total stake of $2N/3$ attested with a checkpoint edge to $(B,j)$ and more than $2N/3$ stake attested with a checkpoint edge to $(B',j)$. These are our desired $\validatorset_1$ and $\validatorset_2$.
\end{proof}

\begin{remark}
The importance of Lemma~\ref{lem:unique-jep} is that assuming reasonable conditions (we are not $(1/3)$-slashable), the choice of $(B_j, j)$ in Algorithm~\ref{alg:hlmd-ghost} is unique and thus well-defined. More generally, we can assume any view only sees at most one justified pair $(B,j)$ per epoch $j$ over \textbf{all} chains.
\end{remark}



\begin{prop}
\label{prop:honest-no-violation}
An honest validator following the protocol will never accidentally violate the slashing conditions.
\end{prop}

\begin{proof}
First, notice an honest validator cannot violate (S1) because each validator is assigned to exactly one committee in each epoch, and is thus asked to attest exactly once per epoch. Thus, it suffices to consider (S2).

Suppose an honest validator $V$ is about to write an attestation that violates (S2) in epoch $u$. This means there are epochs $r < s < t < u$ where in epoch $t$, $V$ wrote an attestation $\alpha_t$ with checkpoint edge $(B_2, s) \votes (B_3, t)$ and the protocol is telling $V$ now to write an attestation $\alpha_u$ with checkpoint edge $(B_1, r) \votes (B_4, u)$. This means running HLMD GHOST on $V$'s current view ends up at some leaf block $B$, which must be a descendent of $B_4$ such that $\lastepochpair(B) = (B_4, u)$.  

Because $V$ wrote $\alpha_t$ in epoch $t < u$, we know $(B_2, s)$ was a justified pair in $\ffgview()$ of some leaf block at that point in epoch $t$. As justified blocks remain justified when their chains grow, we know that now in epoch $u$, the starting justified pair $(B_J, j)$ as defined in HLMD GHOST (Algorithm~\ref{alg:hlmd-ghost}) must satisfy $j \geq s$. Since $s > r$, we know that $j > r$.

However, because of the filtering $L'$ in Algorithm~\ref{alg:hlmd-ghost}, we know that the resulting leaf block $B$ of running HLMD GHOST (in epoch $u$) starting at $(B_J, j)$ must satisfy $(B_J, j) \in J(\ffgview(B)).$ Since $\lastjustpair(B) \in J(\ffgview(B))$ by definition and $\lastjustpair(B) = (B_1,r)$, we know $j \leq r$, a contradiction. Therefore, the protocol cannot force an honest validator to violate (S2).
\end{proof}

\subsection{Rewards and Penalties}

A validator should be rewarded (i.e. have stake increased) for either including valid attestations in his/her proposed block (in the beacon chain specs \cite{beacon}, \emph{proposer reward}) or attesting to the correct block which is justified and finalized as the chain grows (in \cite{beacon}, \emph{attester reward}). Meanwhile, a validator should be penalized (have stake decreased) for violating slashing conditions.

The reward and penalty amounts can be adjusted based on the security level that needs to be achieved, and the game theory of the situation should be such that validators are incentivized to perform their tasks and to not violate slashing conditions. Additionally, it is worth considering more complex incentives, such as incentivizing validators to catch other misbehaving validators. While this makes for potentially interesting game theory work, adding a layer of analysis of these parameters is distracting for the scope of this paper, where we want to focus on the consensus aspect of $\mainprotocol$. To be pragmatic, we abstract this analysis away by assuming these reward and penalty mechanisms provide enough game-theoretical incentives such that the following holds: 
\begin{itemize}
    \item Honest validators follow the protocol.
    \item Any honest validator seeing a slashing condition violated will slash the dishonest validator.
\end{itemize} 

\section{Safety}
\label{sec:safety}



Recall our main goals are proving safety and liveness. In this section, we prove safety, which is done similarly as in Casper FFG \cite{buterin2017casper}; the main differences are that we are using epoch boundary pairs (as opposed to Casper's checkpoint blocks) and our finalization definition is more complex. Thus, we need the following Lemma, which is necessarily more complicated than its counterpart idea in Casper.

\begin{lemma} \label{lem:finalized-no-split} In a view $G$, if $(B_F, f) \in F(G)$ and $(B_J, j) \in J(G)$ with $j > f$, then $B_F$ must be an ancestor of $B_J$, or the blockchain is $(1/3)$-slashable -- specifically, there must exist $2$ subsets $\validatorset_1, \validatorset_2$ of $\validatorset$, each with total stake at least $2N/3$, such that their intersection all violate slashing condition (S1) or all violate slashing condition (S2). 
\end{lemma}
\begin{proof}

Anticipating contradiction, suppose there is a pair $(B_J, j)$ with $j > f$ and $B_J$ is not a descendant of $B_F$. By definition of finalization, in $G$, we must have $(B_F, f) \supervotes (B_k, f+k)$, where we have a sequence of adjacent epoch boundary pairs $(B_F, f), (B_1, f+1), \dots, (B_k, f+k)$.

Since $(B_J,j)$ is justified and $B_J$ is not a descendant of $B_F$, without loss of generality (by going backwards with supermajority links), we can assume $(B_J, j)$ is the earliest such violation, meaning we can assume $(B_l, l) \supervotes (B,j)$ where $l < f$ but $j > f$ (here we are using Lemma~\ref{lem:unique-jep}, which tells us no two justified blocks have the same $\attestepoch()$, else we are done already with $2$ validator subsets of weight $2N/3$ each violating (S1); this is why we do not worry about the equality case). Since $B_1, \ldots, B_k$ are all justified but are descendants of $B_F$, we know $B_J$ cannot be any of these blocks, so we must have $j > f+k$.  
 This means the view $G$ sees some subset $\validatorset_1$ of $\validatorset$ with total stake more than $2N/3$ have made attestations justifying a checkpoint edge $(B_l, l) \rightarrow (B_J, j)$, so for any such attestation $\alpha_1$, $\attestepoch(\lastjustpair(\alpha_1)) = l$ and $\epoch(\alpha_1) = j$. Similarly, $G$ also sees more than $2N/3$ weight worth of validators $\validatorset_2$ have made attestations justifying $(B_F, f) \rightarrow (B_k, f+k)$ so for any such attestation $\alpha_2$, $\attestepoch(\lastjustpair(\alpha_2)) = f$ and $\epoch(\alpha_2) = f+k$. Thus, for anyone in the intersection $\validatorset_1 \cap \validatorset_2$, they have made two distinct attestations $\alpha_1$ of the former type and $\alpha_2$ of the latter type.  Because
 \[
 l < f < f+k < j,
 \]
 we know
 \[
\attestepoch\left(\lastjustpair\left(\alpha_{1}\right)\right)<
\attestepoch\left(\lastjustpair\left(\alpha_{2}\right)\right)<
\epoch\left(\alpha_{2}\right)<
\epoch\left(\alpha_{1}\right),
 \]
which allows them to be provably slashed by (S2).
\end{proof}

We are now ready to prove safety. 
\begin{theorem}[Safety] \label{thm:safety}
In a view $G$, if we do not have the following properties, then $G$ is $(1/3)$-slashable.
\begin{enumerate}
    \item any pair in $F(G)$ stays in $F(G)$ as the view is updated.
    \item if $(B,j) \in F(G)$, then $B$ is in the canonical chain of $G$.
\end{enumerate}
\end{theorem}
\begin{proof}
The first property is straightforward from the definitions of $F(G)$ and $J(G)$, so we omit the proof. 

The definition of Algorithm~\ref{alg:hlmd-ghost} shows that the canonical chain always goes through the justified pair with highest attestation epoch $j$, so by Lemma~\ref{lem:finalized-no-split} must go through the highest finalized pair in $F(G)$. Thus, it suffices to show that no two finalized blocks conflict, because then all the finalized blocks must lie on the same chain, which we just showed must be a subset of the canonical chain.

We now show that if $(B_1, f_1), (B_2, f_2) \in F(G)$ but $B_1$ and $B_2$ conflict, then $G$ is $(1/3)$-slashable. Without loss of generality, $f_2 > f_1$. Since $(B_2, f_2)$ is finalized, it is justified, and we can apply  Lemma~\ref{lem:finalized-no-split}. This shows that either $B_2$ must be a descendant of $B_1$ (assumed to be impossible since they conflict) or $G$ is $(1/3)$-slashable, as desired. 
\end{proof}

Finally, to get our original notion of safety from the Theorem, note that the definition of finalized blocks automatically persist in future views, 

\section{Plausible Liveness}
\label{sec:plausible-liveness}

This section is similar to Casper FFG \cite[Theorem 2]{buterin2017casper}. Since our setup is more complex (in particular, Algorithm~\ref{alg:hlmd-ghost} is much more involved), the results do not immediately follow from the previous result, despite having the same main ideas.

To do this, we give an auxiliary definition: given a chain $c = \chain(B)$ and an epoch $j$, we say that $c$ is \emph{stable at $j$} if $\lastjustpair(\ebb(B, j)) = (B', j-1)$ for some block $B'$.


\begin{theorem}
\label{thm:plausible-liveness}
If at least $2N/3$ stake worth of the validators are honest, then it is always possible for a new block to be finalized with the honest validators continuing to follow the protocol, no matter what happened previously to the blockchain.
\end{theorem}
\begin{proof}





Suppose we are starting epoch $j$. Specifically, suppose our current slot is $i = Cj$; then it is plausible (by having good synchrony) for everyone to have the same view. In particular, the proposer, who is plausibly honest, would then propose a new block $B$ with slot $i$, which is a child of the output block of $\mlmd()$ run on the old view. Call this current view (including $B$) $G$ and define $c = \chain(B)$. Keep in mind that $\lastepoch(B) = B$; we are introducing a new epoch boundary block.

Now, we claim that it is plausible that $c$ extends to a stable chain at the beginning of the next epoch $(j+1)$. To see this, note that since at least $2N/3$ stake worth of the validators are honest, it is plausible that they all attest for $B$ or a descendant (for example, if they are all synced with the network view and vote immediately after $B$ is created). This creates a supermajority link  $\lastjustpair(B) \supervotes (B, j)$, justifying $B$. Now, in the next epoch $(j+1)$, it is plausible for the new block $B'$ with slot $i + C = (j+1)C$ to include all of these attestations (which would happen with good synchrony), so $\lastjustpair(B') = (B, j)$ and $\chain(B')$ is indeed stable at epoch $(j+1)$. Thus, by assuming good synchrony and honest validators, it is plausible (possibly having to wait $1$ extra epoch) that the chain of the first epoch boundary block of the new epoch is stable, no matter what the state of the blockchain was initially.

Thus, we can reduce our analysis to the case that $c$ was stable to begin with, meaning that we have a supermajority link $(B', j-1) \supervotes (B, j)$ in $\ffgview(B)$. Then by the same logic above, it is plausible for the next epoch boundary block ($B''$ with slot $(j+1)C$) to also be stable, with another supermajority link $(B, j) \rightarrow (B'', j+1)$. This finalizes the pair $(B, j)$; in particular, it is the special case of $1$-finalization. 
\end{proof}

The original Casper FFG is a ``finality gadget'' on top of a (presumed probabilistically live) blockchain, so its focus was not on the probabilistic liveness of the underlying chain, rather that the participants do not get into a situation where we are forced to slash honest validators to continue. Since our paper also provides the underlying protocol, we also want a probabilistic guarantee of liveness, which we will cover in Section~\ref{sec:liveness}.


\section{Probabilistic Liveness}
\label{sec:liveness}

In this section, we prove probabilistic liveness of our main protocol $\mainprotocol$ from Section~\ref{sec:stage-2}, given some assumptions. Our proof consists of the following steps, which are basically the main ideas of the proof of Theorem~\ref{thm:plausible-liveness}:
\begin{enumerate}
    \item (assumptions lead to high weight after first slot) Under ``good'' conditions (such as having enough honest validators, synchrony conditions, etc. to be formalized)  honest validators have a high probability of finding a block with a high weight after the first slot.
    \item (high weight after first slot leads to high probability of justification) if we have a block with high weight after one slot, its weight ``differential'' over competitors is likely to increase throughout the epoch, meaning that this block (or a descendent) is likely to be justified.
    \item (high probability of justification leads to high probability of finalization) if every epoch is likely to justify a block, then it is very likely for at least one block to be finalized in a stretch of $n$ epochs as $n$ increases.
\end{enumerate}
This division is not just organizational; it also encapsulates the assumptions into different parts of the argument so that each part would remain true even if assumptions change. For example, the main theorem corresponding to the $3$rd item, Theorem~\ref{thm:finalization-liveness}, does not make any synchrony assumptions and only relies on a probability $p$ that a block gets justified (though synchrony assumptions may be required to bound $p$ given the previous parts). Also, we assume in the second part (as we do for simulations in Appendix~\ref{sec:equivocation-game}) that all validators have equal stake, whereas the third part does not. 

This section is the most intricate and dependent on parameters: for one, the current planned implementation for Ethereum 2.0 actually skirts a nontrivial portion of the analysis due to the attestation consideration delay (see Section~\ref{sec:attestation_consideration_delay}); also, the bounds assume equal stake per validator, which may or may not be a useful assumption depending on external factors. However, it is still worthwhile to include an analysis for the ``pure'' protocol because the proof strategies here, such as using concentration inequalities, are very generalizable to a whole class of these potential slot-based approaches to proof-of-stake, and we believe our analysis captures the heart of what makes these approaches work probabilistically. In fact, ``patches'' such as the aforementioned attestation consideration delay exist primarily because of the types of attacks against the pure protocol, so it is useful for intuition-building to address these attacks even if the actual implementation sidesteps them. Thus, the primary value of this section is as a proof-of-concept of the analysis involved in designing probabilistically-live protocols such as $\mainprotocol$, as opposed to a mathematical guarantee of the actual implementation.

\subsection{High Weight After the First Slot - The Equivocation Game}
\label{sec:liveness-equivocation-game}

We define the \emph{equivocation game} to be the following:

\begin{enumerate}
\item The game is parametrized by $(\validatorset, a, \epsilon_1, \epsilon_2)$, where $\validatorset$ is a set of validators and $(a, \epsilon_1, \epsilon_2)$ are real numbers that parameterize network / synchrony conditions (for ease of reading we push the details to the Appendix).
\item As in $\mainprotocol$, $|\validatorset| = N$ and each $V \in \validatorset$ has some fixed stake $w(V)$. We assume the total stake is $N$ (so the average stake is $1$) and the total amount of stake of honest validators is at least $2N/3$. 
\item There are $2$ options to vote for, which we call $O_1$ and $O_2$ for short. These are abstractions for voting on $2$ conflicting blocks in $\mainprotocol$, and the concept of ``blocks'' are not included in this equivocation game. We (meaning the honest validators) \emph{win} if either $O_1$ or $O_2$  obtain at least $2N/3$ stake worth of votes, and we lose (i.e., we are in a state of equivocation) otherwise. Note that this is equivalent to the idea that there is a stake differential of $N/3$ between the two choices after the game.
\end{enumerate}


The equivocation game is defined to capture the idea of $\mainprotocol$ with just $2$ choices in one slot. Most of the necessary assumptions for liveness are encoded into the game so that the other liveness results in this section can be as independent as possible. 


In Appendix~\ref{sec:equivocation-game}, we flesh out the details of the equivocation game and conduct the simulations of equivocation game with three regimes (a pessimistic regime, an optimistic regime, and a regime in between). The main idea is that under ``reasonable'' conditions $(\validatorset, a, \epsilon_1, \epsilon_2)$, it is very likely that after the first slot, we win the game, meaning that we should expect one of the  pairs in $\mainprotocol$ to have $2N/3$ worth of stake supporting it after the first slot.

For the big picture, the only thing we need for the rest of Probabilistic Liveness is the idea that we ``win'' the equivocation game (specifically, have a block with a high number of attestations over the second-best option after the first slot) with some probability $r$, and having higher $r$ increases the guarantees that a block will be justified in Theorem~\ref{thm:justification-liveness}. For practical purposes, the work in the Appendix seems like an $r$ of around $80\%$ seems to be a reasonable heuristic under not too strong or weak synchrony conditions, and that is already assuming a very bad case where close to $1/3$ of validators are dishonest.

\begin{remark}
Clearly, the concept of equivocation games can be generalized, and similar games could be used to study other proof-of-stake models. There are many directions (the protocols for honest validators, latency modeled by something more than the uniform distribution, etc.) which may be interesting to study for future work, but may be too distracting for the purpose of our paper.
\end{remark}

\subsection{High Weight after First Slot Leads to Justification}
\label{sec:justification-liveness}

In this section, we focus on the $j$-th epoch $[T, T+C]$, where $T=jC$. Let $S = N/C$. We define the following set of assumptions, which we call $A(\epsilon)$\label{Assumptions: liveness}:
\begin{itemize}
    \item the network is $(1/2)-$synchronous starting at time $T = jC$;
    \item each validator has $1$ stake;
    \item the total number of byzantine validators is equal to $N/3-C\epsilon$, meaning the average number of byzantine validators in each slot's committee is $S/3-\epsilon$. 
\end{itemize}

For each $i \in \left\{ 0,1,\ldots,C-1\right\}$, define $h_i$ to be the number of guaranteed honest attestors in slot $jC+i$ and $b_i$ to be the number of dishonest attestors in the slot. Let $S = N/C$, so for all $i$, $h_i + b_i = S$. We show that the distribution of honest / dishonest attestors should not stray too much from expectation in Proposition~\ref{prop:justification-liveness}, using ubiquitous concentration inequalities:

\begin{prop}
\label{prop:serfling}
Let $\mathcal{X} = (x_1,...,x_N)$ be a finite list of $N$ values with $a \le x_i \le b$ and let $X_1,...,X_n$ be sampled without replacement from $\mathcal{X}$. The following hold:

\[
\prob{\sum_{i=1}^{n}X_{i}-n\E{X}\ge n\delta}\le \e{\frac{-2n\delta^2}{(1 - (n-1)/N)(b-a)^2}}
\]

\[\prob{\sum_{i=1}^{n}X_{i}-n\E{X}\le -n\delta}\le \e{\frac{-2n\delta^2}{(1 - (n-1)/N)(b-a)^2}}
\]
\end{prop}

\begin{proof}
See e.g.  \cite{serfling1974}.
\end{proof}

\begin{prop}
\label{prop:justification-liveness}
Suppose the assumptions $A(\epsilon)$  are met. Assume that $C = 2^L$. Then the conjunction of the following events (we purposefully skip $h_0$):
\begin{enumerate}
\item [$E_1$]: $h_1 \geq 2 \cdot \frac{S}{3}$
\item [$E_2$]: $h_2 + h_3 \geq 2^2 \cdot \frac{S}{3}$
\item [$E_3$]: $h_4 + h_5 + h_6 + h_7 \geq 2^3 \cdot \frac{S}{3}$
\item [$\cdots$]: $\cdots$
\item [$E_L$]: $h_{C/2} +\cdots + h_{C-1} \geq 2^{L} \cdot \frac{S}{3}$
\end{enumerate}
has probability
\[
\mathbb{P}\left(\bigcap_{i=1}^{L}E_{i}\right)\geq 1 - \sum_{i=1}^L \e{-\frac{2^i\eps^2}{\left(1 - \frac{2^{i-1}S-1}{2^LS}\right)S}} \geq 1-\sum_{i=1}^{L}\exp\left( -\frac{2^{i}\epsilon^{2}}{S}\right).
\]

\end{prop}





\begin{proof}
 We use the hypergeometric distribution model; that is, we consider the set $\mathcal{X} = (x_1, \cdots, x_{N})$ representing the validators $\mathcal{V}$, where $x_j = 1$ if the $j$-th validator is honest and $0$ otherwise. Let $X_1,\ldots, X_{2^{i-1}S}$ be $2^{i-1}S$ samples from $\mathcal{X}$ without replacement. Then for each $j$, $1\leq j \leq 2^{i-1}S $, $\E{X_j} = 2/3 + \eps / S$, so:
\[
\E{\sum_{j=1}^{2^{i-1}S}X_j} = 2^{i-1}S(2/3+\eps/S) = 2^i S/3 + 2^{i-1}\eps.
\]
We can then bound each $E_i$ (as a sum of $X_j$'s) by
\begin{align*}
\prob{\sum_{j=1}^{2^{i-1}S} X_j \ge 2^i \frac{S}{3}} &= 1 - \prob{\sum_{j=1}^{2^{i-1}S} X_j < 2^i  \frac{S}{3}}\\
&= 1- \prob{\sum_{j=1}^{2^{i-1}S} X_j - \left( \frac{2^iS}{3} + 2^{i-1}\eps \right) < -(2^{i-1}S)\frac{\eps}{S}}\\
&\ge 1 - \e{-\frac{2^i\eps^2}{\left(1 - \frac{2^{i-1}S-1}{2^LS}\right)S}},\\
\end{align*}
 where the last inequality results from Proposition \ref{prop:serfling}, recalling that $N=2^LS$. The probability of each event $E_i$ when considered independently of each other is then bounded below by this value. Using the intersection bound on all $E_i$ we get the first inequality in the desired statement; the second inequality holds by comparing denominators.
\end{proof}

\begin{remark}
 It suffices to also use e.g. Hoeffding's inequality with a binomial model where the validators are assigned as honest or byzantine with replacement. We leave it as an exercise to the reader that this method immediately gets the weaker bound
\[
\mathbb{P}\left(\bigcap_{i=1}^{L}E_{i}\right)\geq1-\sum_{i=1}^{L}\exp\left( -\frac{2^{i}\epsilon^{2}}{S}\right).
\]
While the binomial model is only an approximation, the model with replacement is strictly more mean-reverting than the model without, so the approximation is in the correct direction for us. This could be made rigorous with e.g. coupling methods. We use the weaker bound because it is algebraically cleaner and loses very little compared to the stronger bound.

It can be seen that the exponential terms in the result of proposition \ref{prop:justification-liveness} is fairly small when $\epsilon$ is on the order of $\sqrt{S}$. For example, when $C=64$ (a power of $2$ actually makes Proposition~\ref{prop:justification-liveness} work cleanly, though it is certainly not a crucial part of the bound) and thus $S = 900$, picking $\epsilon = 30$ (meaning we have as many as $17280$ byzantine validators out of $57600$) gives a probability bound of around $85\%$; changing $\epsilon$ to $40$ jumps the probability to around $97\%$.
\end{remark}


\begin{lemma} \label{lem:justification-stays-on-chain}
Suppose the assumptions in $A(\epsilon)$ are met.
 Then, if $\view(\network, T+C)$ justifies a new block $B$ not in $J(\view(\network, T))$, $B$ must be a descendant of the last justified block in $J(\view(\network, T))$.
\end{lemma}

\begin{proof}
Recall that honest attestors wait for $1/2$ time before attesting to their assigned slot. Thus, in the upcoming epoch $j$, all of the honest attestors in this epoch attest in the time period $[T+1/2, T+C]$. Because we have $(1/2)$-synchrony, it means their views during this period all include $\view(\network, T)$. By the nature of Algorithm~\ref{alg:hlmd-ghost}, this means all of their block proposals and attestations must be to either:
\begin{itemize}
\item a descendant of $B_J$, the last justified block in $\view(\network, T)$, or
\item a descendant of some new block justified in epoch $j$.
\end{itemize}
Our Lemma only fails in the second case, and if some new block in this epoch obtains $2/3$ of the attestations of this epoch. 

Luckily, the chicken-and-egg favors us. It is possible for byzantine validators to propose a new block $B_J'$ that's not a descendant of $B_J$. However, it would be impossible for $B_J'$ to receive enough votes in this epoch, as all of the attestations before $B_J'$ is justified must go to a descendant of $B_J$ by Algorithm~\ref{alg:hlmd-ghost}. Thus, we know that if we justify a new block, it must be a descendant of $B_J$.
\end{proof}

\begin{theorem}[Justification Probabilistic Liveness] \label{thm:justification-liveness}
Suppose the assumptions in $A(\epsilon)$ are met, and suppose we win the equivocation game corresponding to the first slot with probability $r$. Let $B_J$ be the last justified block in $J(\view(\network, T))$. Then, $\view(\network, T+C)$ will justify a new descendant of $B_J$ with probability at least
\[
\ensuremath{r-\sum_{i=1}^{L}\exp\left( -\frac{2^{i}\epsilon^{2}}{S}\right) - \frac{1}{3^{C-1}}}.
\]
\end{theorem}

\begin{proof}
We can think of the first slot of this epoch $[T, T+1]$ as an equivocation game with $S$ validators ($h_0$ honest and $b_0$ byzantine) where all of the options are $B_J$ or a descendant\footnote{As stated in the proof of Lemma~\ref{lem:justification-stays-on-chain}, it is possible that there are options that are not, but they will never receive attestations from honest validators, so it would be strictly worse for dishonest validators to create such options.} of $B_J$. Thus, with probability $r$, one of these options receives at least $2S/3$ attestations after slot $jC$. We call this winning block $B_w$ (ideally, $B_w$ is simply just the block $jC$, though this is not necessary). For sake of weighting in Algorithm~\ref{alg:hlmd-ghost}, note that against any other option in the equivocation game, $B_w$'s weight is winning by at least $(2S/3 - S/3) = S/3$. 

By the intersection bound, with probability at least 
\[
\ensuremath{r-\sum_{i=1}^{L}\exp\left( -\frac{2^{i}\epsilon^{2}}{S}\right) }
\]
we also satisfy the events in Proposition~\ref{prop:justification-liveness}. We now show that when these events are satisfied, \textbf{all remaining honest validators in slots $(jC+1), \ldots, (jC+C-1)$ will vote for $B_w$ or a descendant.}

Consider slot $jC+1$. Because of $E_1$, we know $b_1 < S/3$, so even if all $b_1$ potentially byzantine actors conspire to vote for some other option, the weight advantage of $B_w$ cannot be diminished to $0$ by the time the honest validators vote at time $(jC+1.5)$. This means that all honest validators will keep attesting to $B_w$ or a descendant block (an honest proposer in slot $jC+1$ would propose block $jC+1$ as a descendant of $B_w$, for example). By $E_1$, we know that $B_w$ gains at least $2S/3$ weight while a rival block gains at most $S/3$ weight during this slot, which means that the weight differential preferring $B_w$ changes by at least $(2S/3 - S/3) = S/3$. This means by the end of slot $jC+1$, $B_w$ is now winning with at least weight $S/3 + S/3 = 2S/3.$ 

Now consider $E_2$ and the next $2$ slots, $jC+2$ and $jC+3$. Between them, we know $h_2 + h_3 < 2S/3$, so even if all the byzantine actors conspire, they cannot destroy the winning differential of $B_w$, which by the end of these $2$ next slots will be winning with at least weight $(2S/3 - 2S/3 + 4S/3) = 4S/3.$

Inductively, this logic continues for all remaining slots in the epoch (specifically, all honest validators attest to $B_w$ or a descendant. Here the structure of Algorithm~\ref{alg:hlmd-ghost}  is important because while honest attestors may (and probably will) attest to new blocks, the weight of their attestations are added to that of $B_w$ as well.

We have thus concluded that during all remaining slots, $B_w$ accumulates all attestations from honest validators. After slot $jC$, it has at least $2S/3$ votes with probability $r$. Then, we know it picks up weight at least $h_1 + h_2 + \cdots + h_{C-1}$, for a total weight of 
\[
\frac{2}{3}S+\sum_{i=1}^{L}\frac{2^{i}}{3}S=\frac{2}{3}S+\frac{2}{3}\left(2^{L}-1\right)S =\frac{2^{L+1}}{3}S=\frac{2}{3}N, 
\]
so we indeed achieve enough weight for a supermajority link.




To prove $B_J$ is different from $B_w$, we need at least one honest validator to propose a new block on top of $B_J$ before or in the first slot of epoch $\attestepoch(B_w)$. Since $B_J$ receives more than $\frac{2}{3}$ of votes in epoch $\attestepoch(B_j)$, there are more than $\frac{2}{3}$ honest validators in $\attestepoch(B_j)$. The chance of \textbf{no} honest validator proposing a block on top of $B_J$ in $\attestepoch(B_j)$ is then bounded above by $(\frac{1}{3}) ^{C-1}$, which is vanishingly small.
\end{proof}

\subsection{Probabilistic Justification Leads to Probabilistic Finalization}
\label{sec:prob-just-to-prob-final}

The main theorem is the following:

\begin{theorem}[Finalization Probabilistic Liveness] \label{thm:finalization-liveness}
Assume the probability of justifying a block (as in Proposition~\ref{prop:justification-liveness}) is independently $p \geq 1/2$ for each epoch, the probability of failing to finalize a block in the next $n$ epochs approaches $0$ exponentially as a function of $n$.



\end{theorem}

\begin{proof}

Recall that we expect most finalization to be $1$-finalization, when one justified block justifies an adjacent epoch boundary block. For this proof, it suffices to show that the probability of failing to even $1$-finalize approaches $0$.

We consider an epoch a ``success'' if a block is justified in the epoch, and a ``failure'' otherwise. Thus, if $2$ adjacent epochs are ``successful,'' we finalize a block. Therefore, the probability of not getting a ($k=1$) finalization  in $n$ epochs is the probability that no $2$ adjacent epochs out of the next $n$ are successful. Using the independence assumption, this has the probability
\[
\sum_{i\ge\frac{n}{2}}^n \binom{i+1}{n-i} (1-p)^i p^{n-i},
\]
because for a particular $i \geq n/2$, there are $\binom{i+1}{n-i}$ ways to select $i$ failures.

Since $p \geq 0.5$, we know $p \geq (1-p)$, so the sum is bounded above by 
\[
\left( \sum_{i=n}^{\left\lfloor \frac{n}{2}\right\rfloor }\binom{i+1}{n-i}\right) (1-p)^{n/2}p^{n/2}
\]
It is well-known (by e.g. induction) that: 
\[
\sum_{i=n}^{\left\lfloor \frac{n}{2}\right\rfloor }\binom{i+1}{n-i}=\binom{n+1}{0}+\binom{n}{1}+\binom{n-1}{2}+\cdots=F_{n},
\]
the $n$-th Fibonacci number, which we know is of the form
\[
F_{n}= \frac{1}{\sqrt{5}}\left[\frac{\left(1+\sqrt{5}\right)}{2}\right]^{n} -\frac{1}{\sqrt{5}}\left[\frac{\left(1-\sqrt{5}\right)}{2}\right]^{n}
\]

The second term vanishes as $n\xrightarrow{} \infty$, so our desired quantity is bounded above by (times a constant)
\[
\left[\frac{\left(1+\sqrt{5}\right)\sqrt{p\left(1-p\right)}}{2}\right]^{n}
\]
We have the bound $\sqrt{p\left(1-p\right)} \le 1/2$ (by, e.g. AM-GM inequality), so the failure rate is bounded above by
\[
\frac{1}{\sqrt{5}}\left(\frac{1+\sqrt{5}}{4}\right)^{n},
\]
which goes to $0$ as $n\xrightarrow{}\infty$. Finally, recall that we are only looking at $1$-finalization, so the chances of finalization is strictly higher (though in practice the other types of finalization should occur very rarely).

\end{proof}

\begin{figure}[H]
\begin{center}
\begin{tikzpicture}
[edge from parent/.style={draw,latex-},sibling distance=10mm, level distance=10mm]
\node [draw]{root}
	child {node [draw] {S}
		child {node [draw] {F}
		  child {node [draw] {S}
			child {node [draw] {F}
				child {node [draw] {F}}}}}};

\node [draw]at (2,0){root}
	child {node [draw] {S}
		child {node [draw] {F}
		  child {node [draw] {F}
			child {node [draw] {S}
				child {node [draw] {F}}}}}};	

\node [draw]at (4,0){root}
	child {node [draw] {S}
		child {node [draw] {F}
		  child {node [draw] {F}
			child {node [draw] {F}
				child {node [draw] {S}}}}}};				

\node [draw]at (6,0){root}
	child {node [draw] {F}
		child {node [draw] {S}
		  child {node [draw] {F}
			child {node [draw] {S}
				child {node [draw] {F}}}}}};

\node [draw]at (8,0){root}
	child {node [draw] {F}
		child {node [draw] {S}
		  child {node [draw] {F}
			child {node [draw] {F}
				child {node [draw] {S}}}}}};				
				
\node [draw]at (10,0){root}
	child {node [draw] {F}
		child {node [draw] {F}
		  child {node [draw] {S}
			child {node [draw] {F}
				child {node [draw] {S}}}}}};
				
\end{tikzpicture}
\end{center}
    \caption{The cases where we \textbf{fail} to finalize a block in $n=5$ epochs with $3$ failing epochs. ``S'' denotes success and ``F'' denotes failsure (with respect to Proposition~\ref{prop:justification-liveness}.}
    \label{fig:fail_final}
\end{figure}
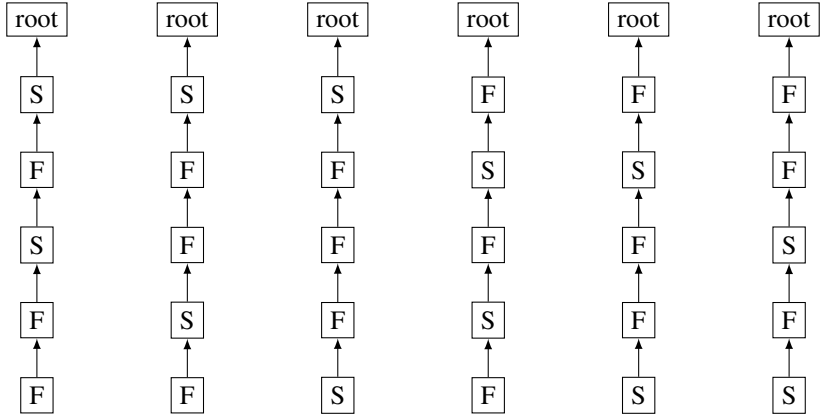

Theorem~\ref{thm:finalization-liveness} is ``agnostic'' of the justification probability $p$ and how we obtained it. This means if we tweak the protocol, change the assumptions, etc. and obtain different bounds/estimates for $p$ than we do from the current Theorem~\ref{thm:justification-liveness}, our result still holds. For this reason, we use $p$ as an \textbf{input} to  Theorem~\ref{thm:finalization-liveness} instead of reusing our actual bound in Theorem~\ref{thm:justification-liveness}. A secondary effect of abstracting away justification liveness is that \textbf{the techniques can be used to prove finalization probabilistic liveness for general $2$-stage protocols, not just $\mainprotocol$} (it could even be extended easily to protocols with $3$ or more stages, though the bounding would be different). Regardless of the protocol, as $n$ increases, as long as $p \geq 0.5$, the probability of finalizing a block increases rapidly, getting around $99\%$ even for $p = 0.5$ after $n=20$ epochs. See Table~\ref{table:no-finalization} for some computations.

\begin{table}[H]
\centering
\caption{The probabilities of \textbf{not} finalizing blocks in $n$ epochs, with $p$ the probability of justification within each epoch.}
\label{table:no-finalization}
\begin{tabular}{|l|l|l|}
\hline
n  & p    & probability of not finalizing any block in n epochs \\ \hline
2  & 0.5  & 0.75                                                \\ \hline
5  & 0.5  & 0.40625                                             \\ \hline
7  & 0.5  & 0.265625                                            \\ \hline
10 & 0.5  & 0.140625                                            \\ \hline
20 & 0.5  & 0.016890525817871094                                \\ \hline
2  & 0.66 & 0.5644                                              \\ \hline
5  & 0.66 & 0.18460210239999997                                 \\ \hline
7  & 0.66 & 0.08322669164799996                                 \\ \hline
10 & 0.66 & 0.025351233503186934                                \\ \hline
20 & 0.66 & 0.0004854107646743359                               \\ \hline
\end{tabular}
\end{table}

\begin{remark}[Relationship with Plausible Liveness] It may feel like we technically get plausible liveness ``for free'' from probabilistic liveness. So why do we treat them separately in this paper? 

For one, plausible liveness is an immediate consequence from the rules of the protocol $\mainprotocol$ and requires very few assumptions, while probabilistic liveness is dependent on probabilistic assumptions (such as network synchrony) and is more fragile with slight changes to the protocol. In particular, our probabilistic liveness treatment makes the very strong assumption that all validators have the same stake. Thus, we choose to present plausible liveness separately to emphasize that even if these assumptions are not satisfied ``in real life,'' plausible liveness remains. 

For another, the emphasis of plausible liveness is that ``honest validators will never be required by the protocol to voluntarily slash themselves to continue'' and the emphasis of probabilistic liveness is ``new blocks will probably be justified/finalized quickly;'' these are different takeaways. 
\end{remark}

\section{Practice versus Theory}
\label{sec:practice-vs-theory}

Ethereum 2.0's practical implementation in~\cite{beacon} contains different design decisions  from $\mainprotocol$, which is meant to be a ``clean'' and more mathematically tractable protocol that captures the theoretical core of the beacon chain design. In this section, we consider the proposed implementation's differences from $\mainprotocol$, without getting lost in the much messier analyses a rigorous study of combining all of these details would require.

\subsection{Sharding}

$\mainprotocol$ is motivated by the Ethereum 2.0 \emph{beacon chain}, which is the ``main'' blockchain in the Ethereum 2.0 design that stores and manages the registry of validators. In this implementation, a \emph{validator} is a registered participant in the beacon chain. Individuals can become a validator of the beacon chain by sending Ether into the Ethereum 1.0 deposit contract. As in $\mainprotocol$, validators create and attest to blocks in the beacon chain. Attestations are simultaneously proof-of-stake votes for a beacon block (as in our design) but also availability votes for a ``shard block,'' which contains data in a different ``shard chain.'' This concept of  \emph{sharding} creates interesting engineering and mathematical questions outside the scope of our paper, which is limited to the beacon chain. 

\subsection{Implementing the View}



In $\mainprotocol$, we can treat views and related concepts as abstract mathematical concepts and validators as perfectly reasoning agents with infinite computational power. In practice, validators will not be directly reasoning with a graphical data structure of a view; instead, they will use software to parse the view given to them and follow the protocol. Thus, in the actual implementation \cite{beacon}, validators run a program that updates the ``store'', which is basically a representation of the view. The store, as the input for LMD GHOST fork choice rule, is updated whenever a block or an attestation is received. The beacon chain also keeps track of a ``state,'' a derived data structure from the view that tracks stake-related data. 

Obviously, mistakes when interpreting these structures and their updates may cause issues with safety and liveness not related to those coming from $\mainprotocol$ itself. Even though in our work we limit our analysis to the mathematical parts of the protocol, we remind the reader these other issues are also important; security holes arising from from a carelessly implemented protocol at the software level are not protected by the mathematical guarantees of $\mainprotocol$.


\subsection{Attestation Inclusion Delay}
\label{sec:attestation_inclusion_delay}

In $\mainprotocol$, when a validator proposes a block, he/she includes all new attestations in his/her view. In the planned implementation, there is an integer parameter for the \emph{attestation inclusion delay} (say $n$) such that when a validator proposes a block, he/she only includes attestations that were made at least $n$ slots ago. 

The purpose of the attestation inclusion delay is to prevent centralization and reward advantage. If the slot length is ``short'' (meaning the latency of attestation propagation is high for normal nodes in the network), then highly-connected nodes might be able to get attestation data faster and be able to publish them faster, which may correspond to various advantages depending on the reward incentives.

The enforced delay allows attestations to disseminate more widely before they can be included in blocks. Thus, nodes have more of an equal opportunity to capture the attestation inclusion rewards for proposing blocks. The plan is to tune the attestation inclusion delay (in the range of $1$ to $4$ slots) depending on real-world network data. Smaller values improve the transaction processing speed, and larger values improve decentralization. Note increasing the slot time is an alternative method to promote decentralization.

\subsection{Attestation Consideration Delay}
\label{sec:attestation_consideration_delay}

In $\mainprotocol$, when a validator is supposed to attest at slot $N$, the validator runs the HLMD GHOST fork-choice rule considering all attestations in the validator's view. In the planned implementation, he/she only considers attestations that are at least $1$ slot old. Specifically, the view used as input to the fork-choice rule contains all valid blocks up to slot $N$ but only contains valid attestations up through slot $N-1$.

This one-slot delay protects validators from a certain class of timing attacks in which byzantine validators eagerly broadcast slot $N$ attestations rather than waiting for the $N+1/2$ time when they are ``supposed'' to attest. This allows them to theoretically take advantage of the network latency to split the vote of the honest validators to keep the chain in a state of equivocation. This is an important attack to counter as publishing attestations at the ``wrong'' time in a block is not a slashable offense (since the adversary can theoretically fake timestamps and we did not assume $\mainprotocol$ has mechanisms that make faking timestamps harder). 

We cover the one-shot delay attack in our analysis of the \emph{equivocation game} (see Appendix~\ref{sec:equivocation-game}), and conclude it is plausible for the attack to be effective in very pessimistic regimes. However, the probability byzantine validators succeed in this particular type of attack is greatly lowered even in pessimistic regimes if we implement something like the attestation consideration delay.

Specifically, this delay improves the parameters for probabilistic liveness as analyzed in Theorem~\ref{thm:justification-liveness}. As long as an honest validator is selected to propose a block (with probability at least $2/3$ in the worst case) $B$, then all honest validators will vote for $B$ in the first slot since they will ignore all the attestations made in the first slot, including those made by the dishonest validators. In other words, the attestation consideration delay neutralizes the $1$-st slot equivocation game as a point of attack, meaning the probabilistic liveness of the proposed implementation comes with even better guarantees than raw $\mainprotocol$. A rigorous analysis of the probabilistic liveness for this setup may make for interesting future work.

\subsection{A Four-Case Finalization Rule}
\label{sec:practice-finalization}

While $\mainprotocol$'s Definition~\ref{def:finalization} captures the ``general'' version of the mathematical idea of finalization, the proposed implementation uses a ``reduced'' version only looking at the last $4$ epochs for practicality. In particular, let $B_1, B_2, B_3, B_4$ be epoch boundary blocks for consecutive epochs, with $B_4$ being the most recent epoch boundary block.


\begin{enumerate}
    
    \item If $B_1, B_2$ and $B_3$ are justified and the attestations $\alpha$ that justified $B_3$ have $LJ(\alpha)=B_1$, we finalize $B_1$.
    
    \item If $B_2, B_3$ are justified and the attestations $\alpha$ that justified $B_3$ have $LJ(\alpha)=B_2$, we finalize $B_2$.
    
    \item If $B_2, B_3$ and $B_4$ are justified and the attestations $\alpha$ that justified $B_4$ have $LJ(\alpha)=B_2$, we finalize $B_2$. 
    
    \item If $B_3, B_4$ are justified and the attestations $\alpha$ that justified $B_4$ have $LJ(\alpha)=B_3$, we finalize $B_3$. 

\end{enumerate}

These are special cases of Definition~\ref{def:finalization} for cases $k=1$ (the first two) and $k=2$. The idea here is Ethereum 2.0 will only honor attestations for up to $2$ epochs. This means only epoch boundary blocks up to $2$ epochs in the past can be newly justified (and thus finalized), which gives the $4$ cases. See Figure~\ref{fig:finalization-real-cases}.

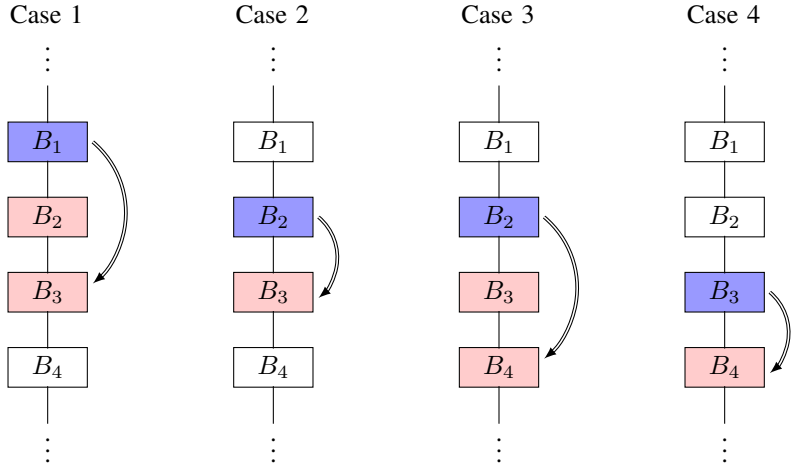
\begin{figure}[H]
\begin{center}

\tikzset {
    basic/.style = {rectangle, minimum width=3em, minimum height = 1.3em},
    root/.style = {basic, thin, align=center},
    block/.style = {draw, basic, thin, align=center}
}
\begin{tikzpicture} [
    level 1/.style={sibling distance = 30mm},
    edge from parent/.style={}, >=latex]

\node  at (-4.5, -0.3)  {\vdots};
\node  at (-1.5, -0.3)  {\vdots};
\node  at (1.5, -0.3)  {\vdots};
\node  at (4.5, -0.3)  {\vdots};

\node [](n) {}
	child {node [block, fill = blue!40](c1){$B_1$}}
    child {node [block](c2){$B_1$}}
    child {node [block](c3){$B_1$}}
    child {node [block](c4){$B_1$}};

\begin{scope}[node/.style={block}]
\node [block, below of = c1, fill = red!20] (c11) {$B_2$};
\node [block, below of = c11, fill = red!20] (c12) {$B_3$};
\node [block, below of = c12] (c13) {$B_4$};
\node [below of = c13] (c14) {$\vdots$};
\end{scope}

\begin{scope}[node/.style={block}]
\node [block, below of = c2, fill = blue!40] (c21) {$B_{2}$};
\node [block, below of = c21, fill = red!20] (c22) {$B_{3}$};
\node [block, below of = c22] (c23) {$B_{4}$};
\node [below of = c23] (c24) {$\vdots$};
\end{scope}

\begin{scope}[node/.style={block}]
\node [block, below of = c3, fill = blue!40] (c31) {$B_{2}$};
\node [block, below of = c31, fill = red!20] (c32) {$B_{3}$};
\node [block, below of = c32, fill = red!20] (c33) {$B_{4}$};
\node [below of = c33] (c34) {$\vdots$};
\end{scope}

\begin{scope}[node/.style={block}]
\node [block, below of = c4] (c41) {$B_{2}$};
\node [block, below of = c41, fill = blue!40] (c42) {$B_{3}$};
\node [block, below of = c42, fill = red!20] (c43) {$B_{4}$};
\node [below of = c43] (c44) {$\vdots$};
\end{scope}


\draw [->,double] (-3.9,-1.5) arc (50:-50:35pt);  
\draw [->,double] (-0.9,-2.5) arc (50:-50:20pt); 
\draw [->,double] (2.1,-2.5) arc (50:-50:35pt);  
\draw [->,double] (5.1,-3.5) arc (50:-50:20pt); 


\draw [-] (-4.5,-.75) -- (-4.5,-1.25);
\draw [-] (-4.5,-1.75) -- (-4.5,-2.25);
\draw [-] (-4.5,-2.75) -- (-4.5,-3.25);
\draw [-] (-4.5,-3.75) -- (-4.5,-4.25);
\draw [-] (-4.5,-4.75) -- (-4.5,-5.25);

\draw [-] (-1.5,-.75) -- (-1.5,-1.25);
\draw [-] (-1.5,-1.75) -- (-1.5,-2.25);
\draw [-] (-1.5,-2.75) -- (-1.5,-3.25);
\draw [-] (-1.5,-3.75) -- (-1.5,-4.25);
\draw [-] (-1.5,-4.75) -- (-1.5,-5.25);

\draw [-] (1.5,-.75) -- (1.5,-1.25);
\draw [-] (1.5,-1.75) -- (1.5,-2.25);
\draw [-] (1.5,-2.75) -- (1.5,-3.25);
\draw [-] (1.5,-3.75) -- (1.5,-4.25);
\draw [-] (1.5,-4.75) -- (1.5,-5.25);

\draw [-] (4.5,-.75) -- (4.5,-1.25);
\draw [-] (4.5,-1.75) -- (4.5,-2.25);
\draw [-] (4.5,-2.75) -- (4.5,-3.25);
\draw [-] (4.5,-3.75) -- (4.5,-4.25);
\draw [-] (4.5,-4.75) -- (4.5,-5.25);

\node[text width = 4cm, font=\fontsize{10}{0}\selectfont] at (-3, 0.2) {Case 1};
\node[text width = 4cm, font=\fontsize{10}{0}\selectfont] at (0, 0.2) {Case 2};
\node[text width = 4cm, font=\fontsize{10}{0}\selectfont] at (3, 0.2) {Case 3};
\node[text width = 4cm, font=\fontsize{10}{0}\selectfont] at (6, 0.2) {Case 4};

\end{tikzpicture}

\end{center}
\caption{\label{fig:finalization-real-cases} From left to right, example of cases $1$ to $4$ respectively. Both red and blue colors indicate justified blocks, where the blue blocks are now finalized because of the observed justification, shown by double arrows.}
\end{figure}


\subsection{Safety - Dynamic Validator Sets}


In $\mainprotocol$'s safety analysis (Section~\ref{sec:safety}), we assumed \emph{static} validator sets, meaning the set of validators cannot change over time. Recall our main result, Theorem~\ref{thm:safety}, tells us we are able to catch $N/3$ weight worth of validators violating the slashing conditions if safety is broken. However, in practice, we would like to support \emph{dynamic} validator sets, which means validators are allowed to \emph{activate} (enter) and \emph{exit} the validator set $V$ over time. This means byzantine activators can act maliciously, but then leave to avoid their stake being slashed. This setup reduces the number of ``active'' validators we can punish for violating slashing conditions.

\begin{define}
Assume there are $2$ validator sets, $\validatorset_1$ and $\validatorset_2$, where $\validatorset_2$ is a validator set later in time compared to $\validatorset_1$. We define $A(\validatorset_1, \validatorset_2$), the validators who \emph{activated (from $\validatorset_1$ to $\validatorset_2$)}, to be the set of validators who are not in $\validatorset_1$ but are in $\validatorset_2$, and $E(\validatorset_1, \validatorset_2)$, the validators who \emph{exited (from $\validatorset_1$ to $\validatorset_2$)}, to be the set of validators who are in $\validatorset_1$ but not in $\validatorset_2$.
\end{define}

First, we note that our key ideas from Section~\ref{sec:safety} still hold:
\begin{lemma} \label{lem:slashing-dvs} Suppose we allow dynamic validator sets. In a view $G$, if $(B_1, f_1)$ and $(B_2, f_2)$ in $F(G)$ conflict, then the blockchain must be $(1/3)$-slashable. Specifically, there must exist $2$ justified pairs $(B_L, j_L)$ and $(B_R, j_R)$ in $G$ and $2$ subsets $\validatorset_1 \subset \validatorset(B_L), \validatorset_2 \subset \validatorset(B_R)$, each with weight at least $2/3$ of the stake of its corresponding attestation epoch, such that their intersection $\validatorset_1 \cap \validatorset_2$ violates (S1) or (S2).
\end{lemma}
\begin{proof}
The proof is essentially identical to that of Lemma~\ref{lem:finalized-no-split} and the reduction of  Theorem~\ref{thm:safety} to it. The only difference is that we must rephrase the conditions of Lemma~\ref{lem:finalized-no-split} as conditions about the validator sets at the time of their attestations.

As in the proof of Theorem~\ref{thm:safety}, if we have $2$ conflicting blocks $(B_1, f_1)$ and $(B_2, f_2)$, then either $f_1 = f_2$ or $f_1 \neq f_2$. If $f_1 = f_2$, then we can set $(B_L, j_L) = (B_1, f_1)$ and $(B_R, j_R) = (B_2, f_2)$ to satisfy our claim, as their intersection violates $(S1)$. If $f_1 \neq f_2$, without loss of generality, $f_1 < f_2$. Since $(B_2, f_2)$ is finalized, it is also justified, and we now satisfy the the setup for Lemma~\ref{lem:finalized-no-split} where $(B_F, f) = (B_1, f_1)$ and $(B_J, j) = (B_2, f_2)$. Applying the Lemma, we conclude that either $B_2$ is a descendent of $B_1$ (a contradiction as they are conflicting), or there must be some pair of pairs (not necessarily $(B_F, f)$ or $(B_J, j)$ themselves), each of whose attestations have $2N/3$ stake, and whose intersections violate (S1) or (S2).
\end{proof}

We now present our main theorem. The main idea is to first take a block $B_0$ (most naturally, we can take the common parent of the conflicting blocks, but this choice is not necessary) published at a ``reference time'' with a ``reference validator set'' $\validatorset_0$. Then, we can upper bound the total weight of activations and exits between the validator sets of the justified blocks $B_L$ and $B_R$ (from Lemma~\ref{lem:slashing-dvs}) and $\validatorset_0$ as a function of the elapsed time since the referenced time. This allows us to lower bound the size of the slashable intersection of the quorums of those blocks.

\begin{theorem}
\label{thm:safety-dynamic} 
Suppose we have a view $G$ which contains two conflicting finalized blocks, then $G$ has enough evidence to slash validators with total weight at least
\[
\max(w(\validatorset_L) - a_L - e_R, w(\validatorset_R) - a_R - e_L) - w(\validatorset_L)/3 - w(\validatorset_R)/3,
\]
where 
\begin{enumerate}
\item there is a block $B_0 \in G$ with validator set $\validatorset_0$;
\item blocks $B_L$ and $B_R$ exist by Lemma~\ref{lem:slashing-dvs}, with validator sets $\validatorset_L$ and $\validatorset_R$ respectively.
\item $a_{L} = w(A(\validatorset_0, \validatorset_L))$ and $e_{L} = w(E(\validatorset_0, \validatorset_L))$; similarly define $a_R$ and $e_R$.
\end{enumerate}
\end{theorem}

\begin{proof}
By assumption, the pair of (unnamed) conflicting finalized blocks give rise to $B_L$ and $B_R$ with  Lemma~\ref{lem:slashing-dvs}, which in turn have quorums $Q_L \subset \validatorset_L$ and $Q_R \subset \validatorset_R$, both with at least $2/3$ of the weight of their corresponding validator sets $V_L$ and $V_R$.

We denote the different sets created by the overlap of our $3$ validator sets to be $A, B, C, D, E, F, G$, as in Figure~\ref{fig:Venn-Diagram-DVS}. 
\begin{figure}[]
\begin{center}
\includegraphics[scale=0.35]{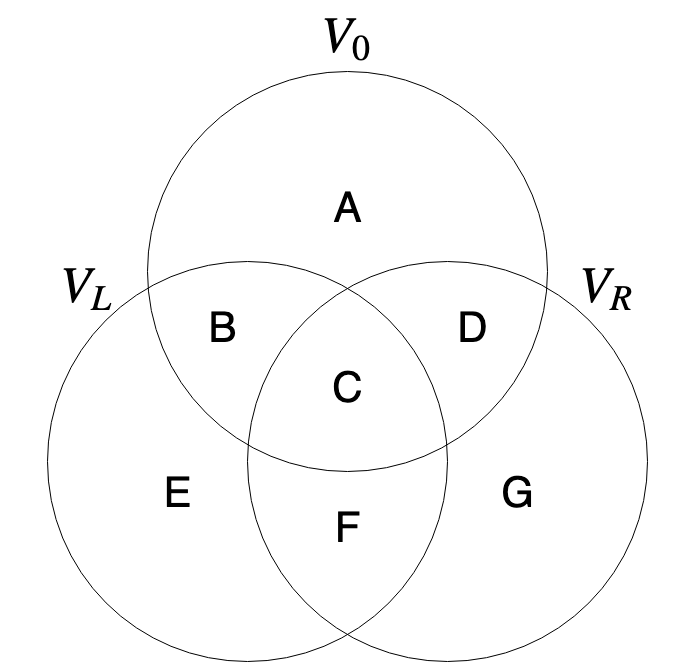}
\end{center}
\caption{A Venn diagram of the initial validator set $\validatorset_0$ and two validator sets $\validatorset_L$ and $\validatorset_R$,  the letters A through G pictorially represent how many validators are in each overlap.}
\label{fig:Venn-Diagram-DVS}
\end{figure}

\begin{figure}[]
\begin{center}
\includegraphics[scale=0.50]{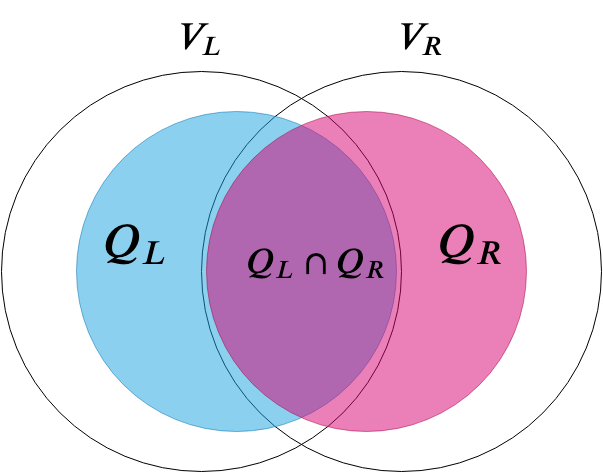}
\end{center}

\caption{The quorums $Q_L$ and $Q_R$ have overlapping validators between the conflicting blocks $B_L$ and $B_R$.}
\label{fig:Venn-Diagram-DVS-Q}
\end{figure}

Next, we wish to bound the intersection of the quorums $Q_L$ and $Q_R$ as depicted in Figure~\ref{fig:Venn-Diagram-DVS-Q}. Let $\validatorset_{LR} = \validatorset_L \cap \validatorset_R = C \cup F$, we have:
\begin{flalign*}
w(Q_L \cap Q_R) &\geq  w(Q_L \cap \validatorset_{LR}) + w(Q_R \cap \validatorset_{LR}) - w(C \cup F) & \\
& \geq   2w(\validatorset_L)/3 - w(B \cup E) +  2w(\validatorset_R)/3 - w(D \cup G \cup C \cup F) & \\
& =  2w(\validatorset_L)/3 + 2w(\validatorset_R)/3 - \left(w(B \cup C \cup D \cup E \cup F \cup G)\right) & \\
& =  w(C) + w(F) - w(\validatorset_L)/3 - w(\validatorset_R)/3. &
\end{flalign*}
Let $X = w(C) + w(F)$. We know
\begin{align*}
X & = w(\validatorset_L) - w(B \cup E) \\
& \geq w(\validatorset_L) - w(A \cup B \cup E \cup F) \\
& = w(\validatorset_L) - a_L - e_R,
\end{align*}
and similarly
\[
X \geq w(\validatorset_R) - a_R - e_L. \qedhere
\]

\end{proof}

As a sanity check, note all the $a$'s and $e$'s are zero in the extreme case of no exits or activations, so we recover the original bound of $w(\validatorset_L)/3 = N/3$. Also, we can avoid negative signs by e.g. using a linear combination of the two bounds for $X$ to get the bound
\[
w(\validatorset_L)/3 - \left(2a_L/3 + 2e_R/3 + a_R/3 + e_L/3\right).
\]

The power of the bound depends on our policies on activating and exiting. Suppose $B_0$, the ``reference block'' was published at time $T_0$, and $B_L$ and $B_R$ are published at $T_{now}$ (more precisely, we should define $T_{now}$ to be the maximum of the two blocks' publication times). Then we can control the $a_*$ and $e_*$ as a function of the time difference and the policies. Here are some examples and observations:
\begin{itemize}
    \item If we allow a constant stake of $k$ worth of new validators to activate per epoch, then we can bound
\[
a_L, a_{R}\le k\left(T_{now}-T_{0}\right).
\]
We can do the same with exiting and obtain bounds on $e_L$ and $e_R$, possibly with a different constant in the same protocol.
\item If we allow activating  proportional to validator stake, e.g., up to $k \cdot w(\validatorset)$ can activate in one epoch, then the bounds become exponential:
\[
a_L, a_R \le w(\validatorset)(1 + k)^{T_{now} - T_0}; 
\]
In the limit this becomes an exponential bound.
\item As the difference between $T_{now}$ and $T_0$ becomes longer, the $w()$ terms will get big and the bound becomes less useful, with the right-hand-side eventually becoming negative. This is unavoidable; one can imagine a ``ship of Thesesus'' situation where the validator sets have changed completely, in which case it would be impossible to slash anyone still in the system.
\item The bounds can also take care of latency; if we are prepared for up to time $\delta$ of latency, then we effectively add $\delta$ to computations where $(T_{now} - T_0)$ occur in our bounds of $a$'s and $e$'s, which we can then feed back into Theorem~\ref{thm:safety-dynamic}.
\end{itemize}

These bounds show \textbf{a tradeoff between the flexibility of entering and exiting the blockchain and the ability to catch malicious actors.} Our static result, Theorem~\ref{thm:safety}, has the intuition ``if our protocol breaks safety, we can provably slash $1/3$ worth of validators'' to imply safeness of our protocol as long as we have a strong enough belief in the validators. Our dynamic result, Theorem~\ref{thm:safety-dynamic}, is of the more nuanced form ``if our protocol breaks safety, then as long as the attackers cannot create a fork with long branch lengths, we can provably slash a little less than $1/3$ worth of validators.'' This means pure $\mainprotocol$ is susceptible to attacks such as malicious actors shutting down the whole network for many\footnote{For reference, it takes roughly $2.5$ months in the concrete protocol to turn over $1/3$ of the validator set, given the queuing mechanism in the Ethereum implementation, which is on the order of $20000$ epochs.} epochs and then suddenly appearing with a conflicting fork, such that we cannot slash much stake worth of validators since they have already exited during this time. Such situations can be safeguarded in practice by simply not including old attestations into one's view; the current Ethereum 2.0 spec, for example, does not accept attestations $2$ or more epochs old, which avoids this attack.


\subsection{Extreme Cases; Hard Forks}

Even though we have fairly unconditional plausible liveness with a $\frac{2}{3}$ majority of honest stakeholders and probabilistic liveness under ``good'' conditions, it is important to take into account of worst-case scenarios. In the case of extended forking and lack of finality due to lots of \emph{non-live} (non-participating) validators who are not necessarily malicious (and thus not slashed), the Ethereum 2.0 beacon chain has a mechanism by which \emph{live} (participating) validators on a fork retain their stake, whereas the non-live validators ``bleed'' stake such that the live validators eventually become a $\frac{2}{3}$ majority, in time on the order of magnitude of weeks. Such a case may happen if, e.g., the global internet is partitioned such that $50\%$ of validators are on one side and $50\%$ on the other. 

A different extreme case is where a chain ``flip-flops'' continuously such that there is no distinct partition, and instead the majority just keeps switching forks each epoch. The ultimate backup plan for this situation or more severe cases is making \emph{hard forks} (or \emph{manual forks}), points at which the community running a blockchain chooses to alter the consensus rules (famous such cases include Bitcoin vs. Bitcoin Cash, or ETH vs. ETH Classic). This might be due to upgrading or adding features, fixing critical issues in production, or addressing a fractured community or political split. Usually, this results in a new set of protocol rules being run starting at some particular point in the blockchain. It is an important engineering problem, again outside of our mathematical scope, that the state transition and fork-choice functions can handle these changes smoothly. In most considerations the HLMD GHOST fork-choice rule would be untouched. However, when the alterations are ``deep'' inside the logic, considerations beyond the built-in forking mechanism would have to be taken and manual forks might be implemented.



\section{Conclusion}
\label{sec:conclusion}

We presented an abstract protocol, $\mainprotocol$, that combines LMD GHOST and Casper FFG for a full proof-of-stake based blockchain design. This is the first formal proof-of-concept of Casper FFG to a complete blockchain protocol. Our goal was to separate the ``mathematically clean'' part of the Ethereum 2.0 design from the implementation details. Some of the analytical techniques are general enough to be useful for studying other protocols, such as the proofs of probabilistic liveness and the equivocation game construction. In practice, much of the worst-case worlds from the analysis can be patched by ideas such as in Section~\ref{sec:practice-vs-theory}.

For transparency, it is important to note that it is possible for ``patches'' or even ``implementation details'' to change safety/liveness assumptions, or even for $2$ otherwise innocuous ``patches'' to combine in undesirable ways. However, it is impossible to predict all the potential implementation changes the protocol will eventually receive, or which of the current proposed patches would be included or removed from implementation for possibly unrelated reasons. Thus, we wrote this paper to analyze the main theoretical foundations of the protocol and we simply sketched an analysis of each implementation detail as an isolated change instead of in conjunction with others. As the protocol solidifies, possibly with real-life testing, it would be meaningful to reassess the effects of patches that would become ``core'' to the protocol. In particular, some version (possibly significantly differently from what we proposed) of dynamic validator sets would be indispensable; once its design is finalized, it would be important to revisit the safety/liveness implications.

There exist many other proof-of-stake based protocols in the cryptocurrency space. Some examples include:

\begin{itemize}
    \item Tendermint \cite{buchman2018latest}: a very ``pure'' design that is a simplification of PBFT applied to a blockchain proof-of-stake context. This design favors safety over liveness. All validators participate in every consensus round and no consensus happens faster than those rounds, with progress only possible if $\geq 2/3$ stake worth of validators are online. 
    \item Casper CBC \cite{caspercbc}: an alternative proposal to Casper FFG, focusing on mathematical correctness by construction. Casper CBC does not have fixed in-protocol thresholds, based on an emergent approach to safety arising from nodes following the majority of what other nodes have done in the past and making finality inferences about blocks. Besides being a proposed protocol, Casper CBC is also meant to be a general framework for analyzing consensus designs. 
    \item Hotstuff (v6)~\cite{yin2018hotstuff}: has many similar properties as Casper FFG, with a flexible mathematical framework; one key difference is that Hotstuff uses an exponential backoff mechanism to be able to make progress under arbitrary network delays. 
    \item Ouroboros~\cite{kiayias2017ouroboros}: a protocol under the ``synchronous'' school which assumes synchrony for fault-tolerance (and thus gets stronger $50\%$ bounds for fault-tolerance). It uses the longest-chain rule and reward mechanisms to incentivize participation and prevent passive attacks. In Ouroboros Praos~\cite{kiayias2017ouroboros}, the protocol is updated in response to vulnerabilities against ``message delay'' attacks. In Ouroboros Genesis in~\cite{badertscher2018ouroboros}, the protocol is further expanded to allow for the ability for a participant to bootstrap the consensus from the genesis block (hence the name), proving globally universally composable (GUC) security against a $50\%$ adversary.
    \item Snow White~\cite{daian2017snow}: a protocol where the committee leader can extend the blockchain with a block that includes a reference to the previous block, similar to block dependencies, and a nonce that can be used to verify the block’s validity against some a priori difficulty constant. Valid timestamps must increase and “any timestamp in the future will cause a chain to be rejected”. Participants only accept incoming chains that have not been modified too far in the past, similar to accepting views within a certain number of epochs.
    \item Nxt \cite{alias2018nxt}:
    a proof-of-stake protocol that has a finite number of tokens. Nxt uses a stake-based probability for the right to generate a block, but it imposes transaction fees to circulate currency since it does not generate new tokens.  Nxt contrasts themselves from Peercoin, noting that Peercoin's algorithm grants more power to miners who have had coins on the network for a long time. The protocol also describes certain restrictions on block creation and transferring tokens to prevent standard attacks and moving stake between accounts.
    \item Thunderella \cite{pass2018thunderella}:
    a proof-of-stake protocol that aims to achieve optimal transaction verification relative to message delay, assuming ``good conditions'' that the leader of a committee is honest and has a super-majority of honest validators. A new leader will be implemented to build upon a back-up network in the event of a stall. The protocol also allows for dynamic validators with cool-down periods.
    \item Dfinity \cite{hanke2018dfinity}: beacon-notarization protocol that works with both proof-of-work and proof-of-stake.  The system uses a random beacon that selects block proposers, and a decentralized notary chooses the highest-ranked block based on a criteria described in \cite{hanke2018dfinity}. Their analogues of \textit{validators} and \textit{committees} are called \textit{replicas} and \textit{groups} respectively, and the decentralized notary depends on the same Byzantine fault-tolerance to notarize blocks and reach consensus. They only have a passive notion of finalization, but the notarization process is quite fast, and blocks do get published in a timely manner. Replicas do need permission to leave the chain.
\end{itemize}

Our goal is not to show that our design is strictly better than any of these other designs. All of these protocols contain tradeoffs based on different design goals and assumptions about the network. Our design is guided by a balance between simplicity, understandability, and practicality, with a mixed emphasis on safety and liveness. For example, if the context of a proposed blockchain prioritizes safety over speed, then one may, e.g., select or modify a proposal more in the direction of Tendermint.

\section*{Acknowledgments}
We thank the Ethereum Foundation and the San Jose State University Research Foundation for support for this program. This project is made possible by the SJSU Math and Statistics Department's CAMCOS (Center for Applied Mathematics, Computation, and Statistics) program.  We also thank Aditya Asgaonkar, Carl Beekhuizen, Dankrad Feist, Brian Gu, Brice Huang, Ryuya Nakamura, Juan Sanchez, Yi Sun, Mayank Varia, and Sebastien Zany for helpful comments. We thank Runtime Verification (Musab Alturki, Elaine Li, Daejun Park) for helpful comments and proof verification with Coq, building off work by Karl Palmskog, Lucas Pena and Brandon Moore. Author Yan X Zhang is the director of CAMCOS;  authors Vitalik Buterin and Danny Ryan are members of the Ethereum Foundation.


\Urlmuskip=0mu plus 1mu\relax
\bibliographystyle{abbrv}
\bibliography{casper}


\newpage

\appendix

\section{Technicalities of Views}
\label{sec:view}

In Section~\ref{sec:setup}, we defined \emph{view} in a streamlined treatment. We provide a more detailed treatment in this section that offers further intuition.

First, we give a different definition from \emph{view} that should seem more intuitive to some readers. We define the \emph{full view} of a validator $V$ at a given time $T$ to be the set of all messages (and timestamps) that have been seen by $V$ by time $T$. We can denote the full view as $\fview(V,T)$. Similarly to the network view $\view(\network, T)$, we also define a ``God's-eye-view'' $\fview(\network, T)$, the \emph{full network view}, as the collection of all messages any validator has broadcast at any time to the network. As with views, for any validator $V$ and any given time $T$, $\fview(\network, T)$ includes all the messages for any $\fview(V, T)$, though the timestamps may be mismatched. With this definition, our definitions $\view(V, T)$ and $\view(\network, T)$ respectively are just subsets of the \textbf{accepted} messages in $\fview(V,T)$ and $\view(\network, T)$ respectively. 

Now, it is completely possible to describe everything in terms of full views instead of views. However, there are many reasons why we work with views instead of full views in this paper (and why we leave the concept of full views to the appendix). To give a few:
\begin{itemize}
    \item We can visualize views as a connected tree of blocks and attestations (with edges corresponding to dependencies). To visualize full views, we necessarily have to use a possibly disconnected graph. 
    \item Motivated by real-life concerns, a protocol working with full views would have to contain instructions that differ on how to handle messages a validator $V$ has seen (but shouldn't act on, as it depends on a possibly nonexistent message $V$ has not seen) versus a message a validator has accepted and fully ``understands''. As an example, a cryptocurrency protocol may have to contain statements of the sort ``when a full view sees a transaction about a coin $B$, if the parent blocks of $B$ exist (recursively) and are all signed correctly, then consider the transaction valid.'' This is very clunky, whereas if parent-child relationships were dependencies then using views abstracts away the recursive-checking logic, and we just need to say ``when a view sees a transaction about a coin $B$, consider the transaction valid.''
    \item It is always possible to have a protocol that works on views: if message $B$ depends on message $A$ but we receive $B$ first, then working with views is equivalent to ignoring the existence of message $B$ when making choices as a validator until $A$ (and other dependencies) come. 
\end{itemize} 

To summarize: even though the \emph{full view} is somewhat intuitively easier to understand than the \emph{view}, since all the dependencies are met in a view by construction, what the view sees forms a ``coherent'' state where the validator can reason about any message with no ambiguity, whereas reasoning about the full view will usually come with caveats about checking that all dependencies of the messages involved in the reasoning are met (probably recursively). 

\section{The Equivocation Game}
\label{sec:equivocation-game}



Recall that in Section~\ref{sec:liveness}, we abstracted a single slot of $\mainprotocol$ into a one-shot game called the \emph{equivocation game}, parametrized by $(\validatorset, a, \epsilon_1, \epsilon_2)$, where $\validatorset$ votes on $2$ options $O_1$ and $O_2$. In this Appendix, we give some further details and perform some simulations.

First, we explicitly discuss how time and synchrony conditions are modelled in this game:
\begin{enumerate}
\item There is a single time period in this game, formalized as the real interval $[0,1]$. This interval corresponds to the $1$ slot of time in $\mainprotocol$.
\item Honest validators follow the following protocol: ``vote at [your] $t=0.5$ for the option with more total stake voted in your view; in the case of a tie, vote $O_1$.'' This protocol corresponds to the instruction in $\mainprotocol$ that attestors for slot $i$ are supposed to attest at time $i + 1/2$ (and tiebreakers are broken by a hash, which is an arbitrary but fixed value).
\item Synchrony model:
\begin{enumerate}
    \item We assume that all validators have perfectly synced clocks, but each validator attempting to vote at time $t$ actually votes at ``real'' time $t + X$ (rounding to inside the interval $[0,1]$), where $X$ is a uniform random variable with support $[-\epsilon_1, \epsilon_1]$, independently re-sampled for each message. We can think of $\epsilon_1$ as a ``timing error'' bound, accounting for clock differences, client-side timing issues, etc. 
    
    \item When a validator votes at ``real'' time $t'$, all other validators obtain the vote at time $t'+a+Y$, where $Y$ is a uniformly distributed variable with support $[-\epsilon_2, \epsilon_2]$.  We can think of $a$ as the average delay per message, and  $\epsilon_2$ a noise on top of the delay. Combining with the previous point, we have that a validator attempting to vote at time $t$ actually has his/her vote received by another validator at $t + a + X + Y$, where the $X$ and $Y$ are independently re-sampled for each event.

\end{enumerate}
\end{enumerate}

For our analysis in this section, we make the following additional assumptions and notations: 
\begin{itemize}
    \item Every validator has exactly $1$ unit of stake. Barring further knowledge (such as empirical data) about validator stake distributions, this is the most intuitive choice for our toy model.
\item We define $N_h \geq 2N/3$ to be the total number of honest validators (which is equivalent to their total amount of stake).  In reality (and in Section~\ref{sec:liveness}), we expect $N_h$ to be a bit bigger at $2N/3 + \epsilon N$ for small $\epsilon$.
    \item Similarly, we define $N_b \leq N/3$ to be the number of byzantine validators. We use $p = \frac{N_b}{N}$ to be the proportion of byzantine validators.
\end{itemize}

Finally, for our computer simulations, we additionally set the parameters to $\validatorset = 111$, $N_h = 74$, $N_b = 37$. The total number of validators of $111$ is based on \cite{Sharding}, as a heuristic lower bound for making the committees safe. Erring on the conservative side, we pick the maximally pessimistic parameter of $p=1/3$ fraction of byzantine validators. Having more honest validators than this assumption significantly improves our chances of winning.

\subsection{The Pessimistic Regime - High Latency}\label{subsec: pessimistic}

This pessimistic regime is a mental experiment to show that under certain conditions fair to both sides, the power of collusion is enough to make the game favor the dishonest side.

We assume that $a$ (the delay for messages from one validator to another) is big compared to $\epsilon_1$ (the error in timing a vote). In this situation, we argue that the dishonest validators can perform the following ``smoke bomb attack'': at time close to $0.5-a$, all the dishonest validators vote in a way that the votes are split between the $2$ options, for a total of $pN/2$ votes for each option, all with fake timestamps for time close to $0.5$. 

When this attack activates, from the view of each honest validator, at time $0.5$ they will see some random votes from the dishonest validators. As $a$ is big compared to $\epsilon_1$, they almost never see other votes from honest validators. Thus, the dishonest validators and honest validators both end up splitting the vote: the dishonest validators end up splitting $pN$ of the votes, and the honest validators split the remaining $(1-p)N$ probabilistically.

Assuming the honest validators all have $1$ stake, the honest validator's votes can be modeled by $(1-p)N$ Bernoulli trials between the two outcomes, ending up close to a normal distribution around a tie, with a standard deviation proportional to $\sqrt{(1-p)N}$ votes. In particular, for large $N$, it is very unlikely for one of the options to get anywhere near $2N/3$ stake worth of votes.

\begin{remark} As an extension of this example, we can also consider the case where the protocol tiebreaker is ``flip a coin,'' the situation becomes even worse: the dishonest validators can simply wait until the end to vote. Because $a$ is big compared to $\epsilon_1$, the honest validators are still essentially making coin flips since they have not seen any other votes. Our tiebreaker rule (vote for option $O_1$ if there is a tie) solves this problem. The main takeaway here is that these protocols (including attacks on them) are very sensitive to very small implementation details. 
\end{remark}

In Figure~\ref{fig:equiv-network-0.3}, we have four cases that demonstrate different outcomes of the pessimistic regime with $a = 0.15$, $\epsilon_1 = 0.05$, and $\epsilon_2 = 0.15$. This ensures that messages sent at time $t$ are received uniformly randomly in the interval $[t, t+0.3]$. Note that having smaller $\epsilon_2$  would be even more pessimistic, almost ensuring that honest validators cannot see each other's votes. 

In Case 1, dishonest validators vote much earlier than honest validators do, so most of the attestations made by the former are visible to the latter when they vote. Thus, the honest validators are likely to vote for the same option. For Cases 2 and 3, dishonest validators vote closer and closer to when the honest validators do, which increases the power of the attack; in Case 3 dishonest validators win at 42\% of the time. In Case 4, dishonest validators vote at the same time as honest validators; because of the delay, the dishonest votes do not confuse the honest validators, so honest validators almost always win (because they are likely to vote for the default option $O_1$, barring seeing any votes).

\tikzset{ 
    table/.style={
        matrix of nodes,
        row sep=-\pgflinewidth, column sep=-\pgflinewidth,
        nodes={
            rectangle,
            draw=black,
            align=center
        },
        minimum height=1.5em,
        text depth=0.5ex,
        text height=2ex,
        nodes in empty cells,
        every even row/.style={
            nodes={fill=gray!20}
        },
        column 1/.style={
            nodes={text width=12em,font=\bfseries}
        },
        row 1/.style={
            nodes={
                fill=black,
                text=white,
                font=\bfseries
            }
        }
    }
}

\begin{figure}[H]
\begin{center}

\begin{tikzpicture} 
\draw (0,0) -- (10,0);

\draw (1.5,0.2) -- (1.5,0);
\draw (2.5,0.2) -- (2.5,0);
\draw (4.5,0.2) -- (4.5,0);
\draw (5.5,0.2) -- (5.5,0);

\draw (0,0.2) -- (0,0);
\draw (10,0.2) -- (10,0);

 \draw[purple,semithick] (2,0) ellipse [x radius=0.5,y radius=.1];
 
 \draw[blue,semithick] (5,0) ellipse [x radius=0.5,y radius=.1];
 
 \draw[orange,thick] (1.5, -0.1) rectangle (5.5, 0.1);

 \node[text width=5, font=\fontsize{8}{0}\selectfont]  at (0,-0.5) {0};
 \node[text width=5, font=\fontsize{8}{0}\selectfont]  at (10,-0.5) {1};
 \node[text width=5, font=\fontsize{8}{0}\selectfont]  at (1.5,-0.5) {0.15};
 \node[text width=5, font=\fontsize{8}{0}\selectfont]  at (2.5,-0.5) {0.25};
  \node[text width=5, font=\fontsize{8}{0}\selectfont]  at (4.5,-0.5) {0.45};
  \node[text width=5, font=\fontsize{8}{0}\selectfont]  at (5.5,-0.5) {0.55};
\end{tikzpicture}

\begin{tikzpicture} 
\draw (0,0) -- (10,0);

\draw (2.5,0.2) -- (2.5,0);
\draw (3.5,0.2) -- (3.5,0);
\draw (4.5,0.2) -- (4.5,0);
\draw (5.5,0.2) -- (5.5,0);
\draw (6.5,0.2) -- (6.5,0);
\draw (0,0.2) -- (0,0);
\draw (10,0.2) -- (10,0);

\draw[purple,semithick] (3,0) ellipse [x radius=0.5,y radius=.1];
 
\draw[blue,semithick] (5,0) ellipse [x radius=0.5,y radius=.1];
 
\draw[orange,thick] (2.5,-0.1) rectangle (6.5, 0.1);

 \node[text width=5, font=\fontsize{8}{0}\selectfont]  at (0,-0.5) {0};
 \node[text width=5, font=\fontsize{8}{0}\selectfont]  at (10,-0.5) {1};
 \node[text width=5, font=\fontsize{8}{0}\selectfont]  at (2.4,-0.5) {0.25};
 \node[text width=5, font=\fontsize{8}{0}\selectfont]  at (3.4,-0.5) {0.35};
 \node[text width=5, font=\fontsize{8}{0}\selectfont]  at (4.4,-0.5) {0.45};
 \node[text width=5, font=\fontsize{8}{0}\selectfont]  at (5.4,-0.5) {0.55};
 \node[text width=5, font=\fontsize{8}{0}\selectfont]  at (6.4,-0.5) {0.65};
\end{tikzpicture}

\begin{tikzpicture} 
\draw (0,0) -- (10,0);

\draw (2.5,0.2) -- (2.5,0);
\draw (3.5,0.2) -- (3.5,0);
\draw (4.5,0.2) -- (4.5,0);
\draw (5.5,0.2) -- (5.5,0);
\draw (6.5,0.2) -- (6.5,0);
\draw (7.5,0.2) -- (7.5,0);
\draw (0,0.2) -- (0,0);
\draw (10,0.2) -- (10,0);

\draw[purple,semithick] (4,0) ellipse [x radius=0.5,y radius=.1];
 
\draw[blue,semithick] (5,0) ellipse [x radius=0.5,y radius=.1];
 
\draw[orange,thick] (3.5,-0.1) rectangle (7.5, 0.1);
 
 \node[text width=5, font=\fontsize{8}{0}\selectfont]  at (0,-0.5) {0};
 \node[text width=5, font=\fontsize{8}{0}\selectfont]  at (10,-0.5) {1};
 \node[text width=5, font=\fontsize{8}{0}\selectfont]  at (2.4,-0.5) {0.25};
 \node[text width=5, font=\fontsize{8}{0}\selectfont]  at (3.4,-0.5) {0.35};
 \node[text width=5, font=\fontsize{8}{0}\selectfont]  at (4.4,-0.5) {0.45};
 \node[text width=5, font=\fontsize{8}{0}\selectfont]  at (5.4,-0.5) {0.55};
 \node[text width=5, font=\fontsize{8}{0}\selectfont]  at (6.4,-0.5) {0.65};
 \node[text width=5, font=\fontsize{8}{0}\selectfont]  at (7.4,-0.5) {0.75}; 
\end{tikzpicture}

\begin{tikzpicture} 
\draw (0,0) -- (10,0);

\draw (2.5,0.2) -- (2.5,0);
\draw (3.5,0.2) -- (3.5,0);
\draw (4.5,0.2) -- (4.5,0);
\draw (5.5,0.2) -- (5.5,0);
\draw (6.5,0.2) -- (6.5,0);
\draw (7.5,0.2) -- (7.5,0);
\draw (8.5,0.2) -- (8.5,0);
\draw (0,0.2) -- (0,0);
\draw (10,0.2) -- (10,0);

 \draw[purple,semithick] (5,0) ellipse [x radius=0.5,y radius=.1];
 
 \draw[blue,semithick] (5,0) ellipse [x radius=0.5,y radius=.1];
 
 \draw[orange,thick] (4.5,-0.1) rectangle (8.5, 0.1);
 
 
 

 
 \node[text width=5, font=\fontsize{8}{0}\selectfont]  at (0,-0.5) {0};
 \node[text width=5, font=\fontsize{8}{0}\selectfont]  at (10,-0.5) {1};
 \node[text width=5, font=\fontsize{8}{0}\selectfont]  at (2.4,-0.5) {0.25};
 \node[text width=5, font=\fontsize{8}{0}\selectfont]  at (3.4,-0.5) {0.35};
 \node[text width=5, font=\fontsize{8}{0}\selectfont]  at (4.4,-0.5) {0.45};
 \node[text width=5, font=\fontsize{8}{0}\selectfont]  at (5.4,-0.5) {0.55};
 \node[text width=5, font=\fontsize{8}{0}\selectfont]  at (6.4,-0.5) {0.65};
 \node[text width=5, font=\fontsize{8}{0}\selectfont]  at (7.4,-0.5) {0.75};
 \node[text width=5, font=\fontsize{8}{0}\selectfont]  at (8.4,-0.5) {0.85};
\end{tikzpicture}

\begin{tikzpicture}

\node[text width=1, font=\fontsize{8}{0}\selectfont, thick, black]  at (-1.4, 1.1) {*};
\draw[purple,semithick] (-.25,1) ellipse [x radius=0.8,y radius=.1];
\node[text width=240, font=\fontsize{8}{0}\selectfont, black]  at (5.1, 1) {: dishonest validators' voting time frame};

\node[text width=1, font=\fontsize{8}{0}\selectfont, thick, black]  at (-1.4, 0.6) {*};
\draw[blue,semithick] (-.25,0.5) ellipse [x radius=0.8,y radius=.1];
\node[text width=240, font=\fontsize{8}{0}\selectfont, black]  at (5.1, 0.5) {: honest validators' voting time frame};

\node[text width=1, font=\fontsize{8}{0}\selectfont, thick, black]  at (-1.4, 0.1) {*};
\draw[orange,thick] (-1.05, -0.1) rectangle (0.55, 0.1);
\node[text width=240, font=\fontsize{8}{0}\selectfont, black]  at (5.1, 0) {: interval of possible time to obtain dishonest validators' votes};

\end{tikzpicture}

\begin{tikzpicture}
\matrix (first) [table,text width=6em]
{
Dishonest voting time & Win \\
$0.2$   & $96\%$ \\
$0.3$   & $74\%$ \\
$0.4$   & $58\%$ \\
$0.5$   & $100\%$ \\
};
\end{tikzpicture}

\caption{Pessimistic regime simulations with $a = 0.15$, $\epsilon_1 = 0.05$, $\epsilon_2 = 0.15$ for the $4$ cases above in order. Each row has a different coordinated dishonest voting time, which affects the winning rates of honest validators. }
    \label{fig:equiv-network-0.3}
\end{center}
\end{figure}
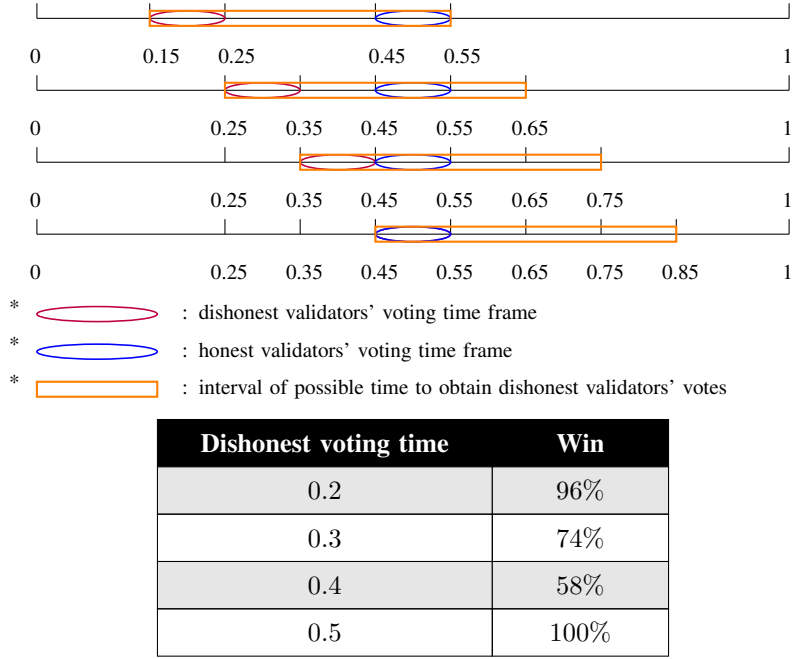

\subsection{The Optimistic Regime - Low Latency}\label{subsec: optimistic}

For this regime, we go to the other extreme and assume perfect synchrony ($a = \epsilon_2 = 0$). This means all decisions are immediately propagated to all other validators on the network. 

In this game, the moment one choice has a lead (by default, $O_1$), everybody is aware of the lead and all honest validators will vote for the choice in the lead. Thus, the optimal strategy for dishonest validators is to keep the voting exactly at a tie (in which case the next honest validator will vote for $O_1$) or such that $O_2$ has one extra vote (in which case the next validator will vote $O_2$). However, as it takes more votes to get the count to under $O_1$, the dishonest validators really can do no better than as if they voted at the very end in the wrong direction, which still gives the honest validators $(1-p)N$ votes in the correct direction. In our situation where $p=1/3$, this gives a total vote differential of $(1-2p)N = N/3$, which is good for us. One can see a simulated outcome of this case in  Figure~\ref{fig:equiv-network-0.0}.

\begin{figure}[H]
\begin{center}
\begin{tikzpicture} 
\draw (0,0) -- (10,0);

\draw (4.9,0.2) -- (4.9,0);
\draw (5.1,0.2) -- (5.1,0);
\draw (0,0.2) -- (0,0);
\draw (10,0.2) -- (10,0);

 \draw[purple,semithick] (5,0) ellipse [x radius=0.1,y radius=.1];
 
 \draw[blue,semithick] (5,0) ellipse [x radius=0.1,y radius=.1];
 
 
 

 
 \node[text width=5, font=\fontsize{8}{0}\selectfont]  at (0,-0.5) {0};
 \node[text width=5, font=\fontsize{8}{0}\selectfont]  at (10,-0.5) {1};
 \node[text width=5, font=\fontsize{8}{0}\selectfont]  at (2.4,-0.5) {0.25};
 \node[text width=5, font=\fontsize{8}{0}\selectfont]  at (3.4,-0.5) {0.35};
 \node[text width=5, font=\fontsize{8}{0}\selectfont]  at (4.4,-0.5) {0.45};
 \node[text width=5, font=\fontsize{8}{0}\selectfont]  at (5.4,-0.5) {0.55};
 \node[text width=5, font=\fontsize{8}{0}\selectfont]  at (6.4,-0.5) {0.65};
 \node[text width=5, font=\fontsize{8}{0}\selectfont]  at (7.4,-0.5) {0.75};
 \node[text width=5, font=\fontsize{8}{0}\selectfont]  at (8.4,-0.5) {0.85};
\end{tikzpicture}

\begin{tikzpicture}
\matrix (first) [table,text width=6em]
{
Dishonest voting time & Win\\
$0.5$   & $100\%$\\
};
\end{tikzpicture}

\caption{Optimistic regime simulation outcome with $a = 0$, $\epsilon_1 = 0.05$, $\epsilon_2 = 0$ with both dishonest and honest validators voting in the same time frame. Honest validators always win.}
    \label{fig:equiv-network-0.0}
\end{center}
\end{figure}
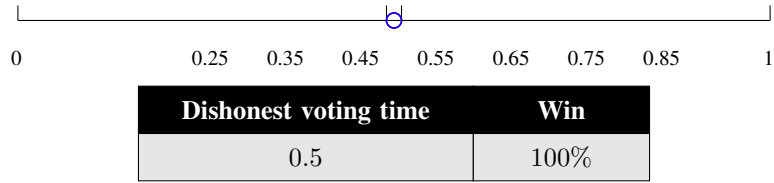

\subsection{An Example Inbetween}

We make a specific set of assumptions somewhere between the extremes presented in sections ~\ref{subsec: pessimistic} and~\ref{subsec: optimistic}. These assumptions are fairly arbitrary and hopefully presents a realistic example. It would be good to find principled ``bottom truth'' parameters somehow (such as after the actual blockchain is implemented).

In Figure~\ref{fig:equiv-network-0.2}, we have three cases that demonstrate the outcomes of the inbetween examples with $a = 0.1, \epsilon_1 = 0.05$ and $\epsilon_2 = 0.1$. Thus, $a$ is comparable to $\epsilon_1$ (avoiding the pessimistic case) but fairly big (avoiding the optimistic case). Making $\epsilon_2$ bigger doesn't really change the results much until $\epsilon_2$ becomes around the order of $a$.

In Case 1, some dishonest validators vote earlier than honest validators do and their ``smoke bomb attack'' has some impact; however, the amount of votes are not distracting enough compared to the amount of votes honest validators are able to vote following the protocol. In case 2, as we earlier discussed in the pessimistic regime, dishonest validators vote closer to when the honest validators do, which boosts the effectiveness of the attack. In Case 3, dishonest validators vote at the same time as honest validators do; as in the pessimistic regime, honest validators still almost always win.

\begin{figure}[H]
\begin{center}

\begin{tikzpicture} 
\draw (0,0) -- (10,0);

\draw (2.5,0.2) -- (2.5,0);
\draw (3.5,0.2) -- (3.5,0);
\draw (4.5,0.2) -- (4.5,0);
\draw (5.5,0.2) -- (5.5,0);
\draw (0,0.2) -- (0,0);
\draw (10,0.2) -- (10,0);

 \draw[purple,semithick] (3,0) ellipse [x radius=0.5,y radius=.1];
 
 \draw[blue,semithick] (5,0) ellipse [x radius=0.5,y radius=.1];
 
 \draw[orange,thick] (2.5, -0.1) rectangle (5.5, 0.1);

 \node[text width=5, font=\fontsize{8}{0}\selectfont]  at (0,-0.5) {0};
 \node[text width=5, font=\fontsize{8}{0}\selectfont]  at (10,-0.5) {1};
 \node[text width=5, font=\fontsize{8}{0}\selectfont]  at (2.4,-0.5) {0.25};
  \node[text width=5, font=\fontsize{8}{0}\selectfont]  at (3.4,-0.5) {0.35};
  \node[text width=5, font=\fontsize{8}{0}\selectfont]  at (4.4,-0.5) {0.45};
  \node[text width=5, font=\fontsize{8}{0}\selectfont]  at (5.4,-0.5) {0.55};
\end{tikzpicture}

\begin{tikzpicture} 
\draw (0,0) -- (10,0);

\draw (3.5,0.2) -- (3.5,0);
\draw (4.5,0.2) -- (4.5,0);
\draw (5.5,0.2) -- (5.5,0);
\draw (6.5,0.2) -- (6.5,0);
\draw (0,0.2) -- (0,0);
\draw (10,0.2) -- (10,0);

\draw[purple,semithick] (4,0) ellipse [x radius=0.5,y radius=.1];
 
\draw[blue,semithick] (5,0) ellipse [x radius=0.5,y radius=.1];
 
\draw[orange,thick] (3.5,-0.1) rectangle (6.5, 0.1);
 
 \node[text width=5, font=\fontsize{8}{0}\selectfont]  at (0,-0.5) {0};
 \node[text width=5, font=\fontsize{8}{0}\selectfont]  at (10,-0.5) {1};
 \node[text width=5, font=\fontsize{8}{0}\selectfont]  at (3.4,-0.5) {0.35};
 \node[text width=5, font=\fontsize{8}{0}\selectfont]  at (4.4,-0.5) {0.45};
 \node[text width=5, font=\fontsize{8}{0}\selectfont]  at (5.4,-0.5) {0.55};
 \node[text width=5, font=\fontsize{8}{0}\selectfont]  at (6.4,-0.5) {0.65};
\end{tikzpicture}

\begin{tikzpicture} 
\draw (0,0) -- (10,0);

\draw (4.5,0.2) -- (4.5,0);
\draw (5.5,0.2) -- (5.5,0);
\draw (6.5,0.2) -- (6.5,0);
\draw (7.5,0.2) -- (7.5,0);
\draw (10,0.2) -- (10,0);

 \draw[purple,semithick] (5,0) ellipse [x radius=0.5,y radius=.1];
 
 \draw[blue,semithick] (5,0) ellipse [x radius=0.5,y radius=.1];
 
 \draw[orange,thick] (4.5,-0.1) rectangle (7.5, 0.1);

 \node[text width=5, font=\fontsize{8}{0}\selectfont]  at (0,-0.5) {0};
 \node[text width=5, font=\fontsize{8}{0}\selectfont]  at (10,-0.5) {1};
 \node[text width=5, font=\fontsize{8}{0}\selectfont]  at (4.4,-0.5) {0.45};
 \node[text width=5, font=\fontsize{8}{0}\selectfont]  at (5.4,-0.5) {0.55};
 \node[text width=5, font=\fontsize{8}{0}\selectfont]  at (6.4,-0.5) {0.65};
 \node[text width=5, font=\fontsize{8}{0}\selectfont]  at (7.4,-0.5) {0.75};
\end{tikzpicture}

\begin{tikzpicture}
\matrix (first) [table,text width=6em]
{
Dishonest voting time & Win\\
$0.3$   & $93\%$\\
$0.4$   & $79\%$\\
$0.5$   & $99\%$\\
};
\end{tikzpicture}

\caption{Simulation outcomes with $a = 0.1$, $\epsilon_1 = 0.05$, $\epsilon_2 = 0.1$. }
    \label{fig:equiv-network-0.2}
\end{center}
\end{figure}
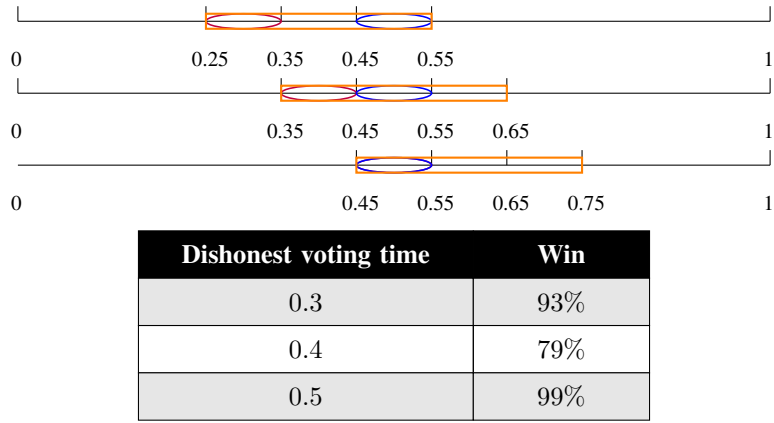
\end{document}